\documentclass[9pt,sigconf]{acmart}

\renewcommand\footnotetextcopyrightpermission[1]{} %

\usepackage{mathtools,amsmath,amsfonts}
\usepackage{amsthm}
\usepackage{thm-restate}
\usepackage{stmaryrd}
\usepackage{xspace}

\usepackage{algorithm}
\usepackage[noend]{algpseudocode}
\usepackage{changepage}
\usepackage{paralist}
\usepackage{thm-restate}
\usepackage{multirow}
\usepackage[capitalise]{cleveref}
\usepackage{graphicx}
\usepackage{subcaption}
\usepackage{longtable}

\usepackage[
colorinlistoftodos, %
textwidth=\marginparwidth, %
textsize=scriptsize, %
]{todonotes}

\theoremstyle{definition}
\newtheorem{theorem}{Theorem}
\newtheorem{corollary}[theorem]{Corollary}
\newtheorem{proposition}[theorem]{Proposition}

\theoremstyle{definition}
\newtheorem{definition}[theorem]{Definition}
\newtheorem{example}[theorem]{Example}
\newtheorem{remark}[theorem]{Remark}
\newtheorem{problem}[theorem]{Problem}

\usepackage{mdframed}

\usepackage{pgf}
\usepackage{tikz}
\usetikzlibrary{shapes,shapes.geometric,fit,trees,decorations,arrows,automata,shadows,positioning,plotmarks,backgrounds,tikzmark}
\usetikzlibrary{calc,matrix,fit,petri,decorations.markings,decorations.pathmorphing,patterns,intersections,decorations.text}

\usepackage{graphicx}
\usepackage[utf8x]{inputenc}
\usepackage[american]{babel}
\usepackage[babel]{csquotes}

\usepackage{algorithm, algpseudocode}

\usepackage{amscd}
\usepackage{mathtools}
\usepackage{array}

\usepackage{xspace}
\usepackage{paralist}
\usepackage{xifthen}
\usepackage{url}

\usepackage{pgf}
\usepackage{tikz}

\newcommand{\colred}{\color{red!60}}

\colorlet{darkgreen}{green!80!black}
\colorlet{darkred}{red!80!black}
\usetikzlibrary{arrows, automata, shapes}
\tikzset{auto, >= stealth}
\tikzset{every edge/.append style={thick, shorten >= 1pt}}
\tikzset{initial/.style={draw, thick, <-, shorten <=1pt}}
\tikzset{player0/.style = {draw, thick, shape=circle, minimum size=5mm}}
\tikzset{player1/.style = {draw, thick, shape=rectangle, minimum size=5mm}}
\tikzset{bplayer0/.style = {draw, thick, shape=ellipse, minimum size=5mm,text width=1.1cm}}
\tikzset{bplayer1/.style = {draw, thick, shape=rectangle, minimum size=5mm,text width=1.6cm}}
\newcommand\pos{1.4}

\usetikzlibrary{shapes.geometric, arrows}

\tikzstyle{startstop} = [rectangle, rounded corners, 
minimum width=0.1cm, 
minimum height=0.5cm,
text centered, 
draw=black, 
]

\tikzstyle{process} = [trapezium, 
trapezium stretches=true, %
trapezium left angle=70, 
trapezium right angle=110, 
minimum width=0.1cm, 
minimum height=0.5cm, 
draw=black, 
]

\tikzstyle{io} = [rectangle, 
minimum width=0.1cm, 
minimum height=0.5cm, 
text centered, 
draw=black, 
]

\tikzstyle{decision} = [diamond, 
aspect=3,
text centered, 
draw=black, 
]
\tikzstyle{arrow} = [thick,->,>=stealth]

\newcommand{\N}{\mathbb{N}}

\newcommand{\bigO}{\mathcal{O}}
\newcommand{\cupdot}{\mathbin{\mathaccent\cdot\cup}}

\newcommand{\Contract}{\mathtt{C}}
\newcommand{\Assumption}{\mathtt{A}}
\newcommand{\AssumptionI}[1]{\Assumption_{#1}}
\newcommand{\Guarantee}{\mathtt{G}}
\newcommand{\GuaranteeI}[1]{\Guarantee_{#1}}

\newcommand{\inodd}{\in_{\mathrm{odd}}}
\newcommand{\ineven}{\in_{\mathrm{even}}}
\newcommand{\abs}[1]{\left\lvert #1 \right\rvert}

\newcommand{\LTLnext}{\bigcirc}
\newcommand{\LTLeventually}{\lozenge}
\newcommand{\LTLalways}{\square}

\newcommand{\lang}{\mathcal{L}}

\newcommand{\game}{\mathcal{G}}
\newcommand{\gamegraph}{\ensuremath{G}}
\newcommand{\play}{\ensuremath{\rho}}
\newcommand{\priority}{\mathbb{P}}
\newcommand{\priorityI}[1]{\priority_{#1}}
\newcommand{\priorityset}[1]{P^{#1}}

\newcommand{\spec}{\ensuremath{\Phi}}
\newcommand{\specI}[1]{\ensuremath{\spec_{#1}}}
\newcommand{\specz}{\ensuremath{\specI{0}}}
\newcommand{\speco}{\ensuremath{\specI{1}}}
\newcommand{\speci}{\ensuremath{\specI{i}}}
\newcommand{\speczp}{\ensuremath{\specz^\bullet}}
\newcommand{\specop}{\ensuremath{\speco^\bullet}}
\newcommand{\specip}{\ensuremath{\speci^\bullet}}
\newcommand{\speczq}{\ensuremath{\specz''}}
\newcommand{\specoq}{\ensuremath{\speco''}}
\newcommand{\speciq}{\ensuremath{\speci''}}

\newcommand{\speczr}{\ensuremath{\specz'}}
\newcommand{\specor}{\ensuremath{\speco'}}
\newcommand{\specir}{\ensuremath{\speci'}}

\newcommand{\vertex}{V}
\newcommand{\vertexI}[1]{\vertex_{#1}}
\newcommand{\vertexz}{\vertexI{0}}
\newcommand{\vertexo}{\vertexI{1}}
\newcommand{\vertexi}{\vertexI{i}}
\newcommand{\edge}{E}

\newcommand{\p}[1]{\ensuremath{\text{Player}~#1}}
\newcommand{\pz}{\ensuremath{\p{0}}}
\newcommand{\po}{\ensuremath{\p{1}}}

\newcommand{\win}{\mathcal{W}}
\newcommand{\winz}{\win_0}
\newcommand{\wino}{\win_1}
\newcommand{\wini}{\win_i}

\newcommand{\strat}{\ensuremath{\pi}}
\newcommand{\stratI}[1]{\ensuremath{\strat_{#1}}}

\newcommand{\stratz}{\ensuremath{\stratI{0}}}
\newcommand{\strato}{\ensuremath{\stratI{1}}}
\newcommand{\strati}{\ensuremath{\stratI{i}}}
\newcommand{\stratj}{\ensuremath{\stratI{j}}}
\newcommand{\Strat}{\ensuremath{\Pi}}
\newcommand{\StratI}[1]{\ensuremath{\Strat_{#1}}}

\newcommand{\Stratz}{\ensuremath{\StratI{0}}}
\newcommand{\Strato}{\ensuremath{\StratI{1}}}
\newcommand{\Strati}{\ensuremath{\StratI{i}}}

\newcommand{\team}[1]{\llangle #1\rrangle}

\newcommand{\buchi}{\ifmmode B\ddot{u}chi \else B\"uchi \fi}
\newcommand{\cobuchi}{\ifmmode co\text{-}B\ddot{u}chi \else co-B\"uchi \fi}
\newcommand{\irmac}{\textsf{iRmaC}\xspace}
\newcommand{\ir}{\textsf{iR}\xspace}
\newcommand{\csm}{\texttt{CSM}\xspace}
\newcommand{\csms}{\texttt{CSMs}\xspace}

\newcommand{\paritygame}{\mathit{Parity}}
\newcommand{\buchigame}{\mathit{B\ddot{u}chi}}
\newcommand{\cobuchigame}{\mathit{co\text{-}B\ddot{u}chi}}

\newcommand{\parity}{\ensuremath{\mathit{Parity}}}
\newcommand{\parityAssump}{\ensuremath{\mathit{Parity}}}

\newcommand{\solveParity}{\textsc{Parity}}
\newcommand{\solveBuchi}{\textsc{B\"uchi}}
\newcommand{\solveCobuchi}{\textsc{CoB\"uchi}}

\newcommand{\solveSafety}{\textsc{Safety}}

\newcommand{\computeAttr}{\textsc{Attr}}
\newcommand{\computeLive}{\textsc{Live}}
\newcommand{\computeCoLive}{\textsc{CoLive}}
\newcommand{\computeSafe}{\textsc{Unsafe}}

\newcommand{\buchiTemp}{\textsc{B\"uchiTemp}}
\newcommand{\cobuchiTemp}{\textsc{coB\"uchiTemp}}
\newcommand{\parityTemp}{\textsc{ParityTemp}}

\newcommand{\template}{\ensuremath{\Lambda}}

\newcommand{\assump}{\ensuremath{\Psi}}
\newcommand{\assumpI}[1]{\ensuremath{\assump_{#1}}}

\newcommand{\assumpz}{\ensuremath{\assumpI{0}}}
\newcommand{\assumpo}{\ensuremath{\assumpI{1}}}
\newcommand{\assumpi}{\ensuremath{\assumpI{i}}}

\newcommand{\templatesafe}{\ensuremath{\template_{\textsc{unsafe}}}}

\newcommand{\templategrlive}{\ensuremath{\template_{\textsc{live}}}}
\newcommand{\templatecondlive}{\ensuremath{\template_{\textsc{cond}}}}
\newcommand{\templatecolive}{\ensuremath{\template_{\textsc{colive}}}}

\newcommand{\src}{\mathit{src}}
\newcommand{\livegroup}{H_\ell}
\newcommand{\livegroupSingle}{H_i}
\newcommand{\livegroupSingleN}{H}
\newcommand{\colivegroup}{D}
\newcommand{\safegroup}{S}
\newcommand{\condlivegroup}{\mathcal{H}}

\newcommand{\livegroupA}{\livegroup^a}
\newcommand{\colivegroupA}{\colivegroup^a}
\newcommand{\safegroupA}{\safegroup^a}
\newcommand{\condlivegroupA}{\condlivegroup^a}
\newcommand{\livegroupS}{\livegroup^s}
\newcommand{\colivegroupS}{\colivegroup^s}
\newcommand{\safegroupS}{\safegroup^s}
\newcommand{\condlivegroupS}{\condlivegroup^s}

\newcommand{\compt}{\lhd}

\newcommand{\pre}[2]{\textsf{pre}\ensuremath{_{#1}(#2)}}

\newcommand{\cpre}[3]{\textsf{cpre}\ensuremath{^{#3}_{#1}(#2)}} %
\newcommand{\cprea}[2]{\textsf{cpre}\ensuremath{^a_{#1}(#2)}}

\newcommand{\attr}[3]{\textsf{attr}\ensuremath{^{#3}_{#1}(#2)}}

\newcommand{\attra}[2]{\textsf{attr}\ensuremath{^a_{#1}(#2)}}

\makeatletter
\newsavebox{\@brx}
\newcommand{\llangle}[1][]{\savebox{\@brx}{\(\m@th{#1\langle}\)}%
	\mathopen{\copy\@brx\kern-0.5\wd\@brx\usebox{\@brx}}}
\newcommand{\rrangle}[1][]{\savebox{\@brx}{\(\m@th{#1\rangle}\)}%
	\mathclose{\copy\@brx\kern-0.5\wd\@brx\usebox{\@brx}}}

\makeatother

\makeatletter 
\newif\ifFIRST
\newif\ifSECOND
\let\LISTOP\relax
\newcommand{\List}[4][\;]{#3#1%
	\FIRSTtrue
	\@for\i:=#2\do{%
		\ifFIRST\LISTOP{\i}\FIRSTfalse\else,\LISTOP{\i}\fi%
	}%
	#1#4%
	\let\LISTOP\relax
}
\makeatother

\makeatletter 

\newcommand{\set}[1]{\left\lbrace #1\right\rbrace}
\newcommand{\tup}[1]{\left( #1\right)}
\newcommand{\Set}[2][]{\List[#1]{#2}{\{}{\}}}
\newcommand{\VSet}[2][]{\let\LISTOP\val\List[#1]{#2}{\{}{\}}}

\newcommand{\Tuple}[2][]{\List[#1]{#2}{(}{)}}

\newcommand{\VTuple}[2][]{\let\LISTOP\val\List[#1]{#2}{(}{)}}

\newcommand{\checkTemplate}{\textsc{CheckTemplate}}
\newcommand{\negotiate}{\textsc{Negotiate}}
\newcommand{\conflict}{\mathcal{C}}
\newcommand{\extract}{\textsc{ExtractStrategy}}

\newcommand{\toolname}{\texttt{CoSMo}\xspace}
\newcommand{\genziel}{\texttt{genZiel}\xspace}

\newcommand{\ro}{{\colbluealt \ensuremath{\mathcal{R}_1}}}
\newcommand{\rt}{{\colorangealt \ensuremath{\mathcal{R}_2}}}

\definecolor{colorblindgreen}{HTML}{004D40}
\definecolor{colorblindblue}{HTML}{1E88E5}
\definecolor{colorblindorange}{HTML}{FFC107}
\definecolor{colorblindred}{HTML}{D81B60}
\definecolor{colorblindbluealt}{HTML}{00BFFF}
\definecolor{colorblindorangealt}{HTML}{FF7F00}

\newcommand{\colbluealt}{\color{colorblindbluealt}}
\newcommand{\colorangealt}{\color{colorblindorangealt}}

\graphicspath{{./img/}{./img/svg/}}
 
\begin{document}
 
\title[]{Contract-Based Distributed Synthesis\\ in Two-Objective Parity Games}

\author{Ashwani Anand}
\affiliation{%
	\institution{MPI-SWS, Germany}
}
\email{ashwani@mpi-sws.org}

\author{Satya Prakash Nayak}
\affiliation{%
	\institution{MPI-SWS, Germany}
}
	\email{sanayak@mpi-sws.org}

\author{Anne-Kathrin Schmuck}
\affiliation{%
	\institution{MPI-SWS, Germany}
}
	\email{akschmuck@mpi-sws.org}

\thanks{
The authors are supported by the DFG projects SCHM 3541/1-1
and 389792660 TRR 248–CPEC. 
}

\begin{abstract}
We consider the problem of computing distributed logical controllers for two interacting system components via a novel sound and complete contract-based synthesis framework. Based on a discrete abstraction of component interactions as a two-player game over a finite graph and specifications for both components given as $\omega$-regular (e.g.\ LTL) properties over this graph, we co-synthesize contract and controller candidates locally for each component and propose a \emph{negotiation} mechanism which iteratively refines these candidates until a solution to the given distributed synthesis problem is found. Our framework relies on the recently introduced concept of \emph{permissive templates} which collect an infinite number of controller candidates in a concise data structure. We utilize the efficient computability, adaptability and compositionality of such templates to obtain an efficient, yet sound and complete negotiation framework for contract-based distributed logical control. We showcase the superior performance of our approach by comparing our prototype tool \toolname to the state-of-the-art tool on a robot motion planning benchmark suite. 
\end{abstract}

\maketitle
\pagestyle{empty}

\sloppy

 	\section{Introduction}\label{sec:intro}
\emph{Games on graphs} provide an effective way to formalize synthesis problems in the context of correct-by-construction cyber-physical systems (CPS) design. A prime example are algorithms to synthesize \emph{control software} that ensures the satisfaction of \emph{logical specifications} under the presence of an external environment, which e.g., causes changed task assignments, transient operating conditions, or unavoidable interactions with other system components. The resulting \emph{logical control software} typically forms a higher layer of the control software stack. The details of the underlying physical dynamics and actuation are then abstracted away into the structure of the game graph utilized for synthesis.

Algorithmically, the outlined control design procedure via games-on-graphs utilizes \emph{reactive synthesis}, a well understood and highly automated design flow originating from the formal methods community rooted in computer science. The strength of reactive synthesis in logical control design is its ability to provide strong correctness guarantees by treating the environment as fully adversarial. 
While this view is useful if a single controller is designed for a system which needs to obey the specification in an \emph{unknown} environment, it does not excel at synthesizing \emph{distributed} and \emph{interacting} logical control software. While controllers for multiple interacting systems can be obtained using existing centralized synthesis techniques, it requires sharing the specifications of all the systems with a central entity. Towards increasing privacy, a decentralized computation of the controllers is preferred, to avoid centrally handling the specifications and strategy choices of the subsystems. In the multi-system setting, each component acts as the \enquote{environment} for the other ones and controllers for components are designed concurrently. Hence, if known a-priory, the control design of one component could take the needs of other components into account and does not need to be treated fully adversarial.

\begin{example}\label{exp:pens:1}
As a simple motivating example for a distributed logical control problem with this flavor, consider a fully automated factory producing pens as depicted in \cref{fig:factory}. It has a machine which takes raw materials for pens at $ A_1 $. When required, it can produce pens with erasers, for which it needs erasers from $ C_1 $. For this, it has a robot $ \ro $ that takes the raw materials from $ B_3 $ to the production machine at $ A_1 $. Hence, the robot $ \ro $ needs to visit $ A_1 $ and $ B_3 $ infinitely often, i.e. satisfy the LTL objective\footnote{We formally introduce LTL objectives in \cref{sec:prelim}.} $\varphi_1:=\square\lozenge \ro:A_1\wedge\square\lozenge \ro:B_3 $, where $ \mathcal{R}_i : P$ denotes that $ \mathcal{R}_i $ is in the cell $ P $. For delivering the erasers to the machine, it has another robot $ \rt $ that takes raw material from $ B_3 $ and feeds the machines via a conveyor belt at $ C_1 $ if $ \ro $ feeds the raw material at $ A_1 $, i.e, the objective is to satisfy $\varphi_2:= \square\lozenge \ro:A_1 \Rightarrow\square\lozenge \rt:B_3\wedge\square\lozenge \rt:C_1 $. As both robots share the same workspace, their controllers need to ensure the specification despite the movements of the other robot.

\begin{figure}
	\centering
		\includegraphics[width=0.45\columnwidth]{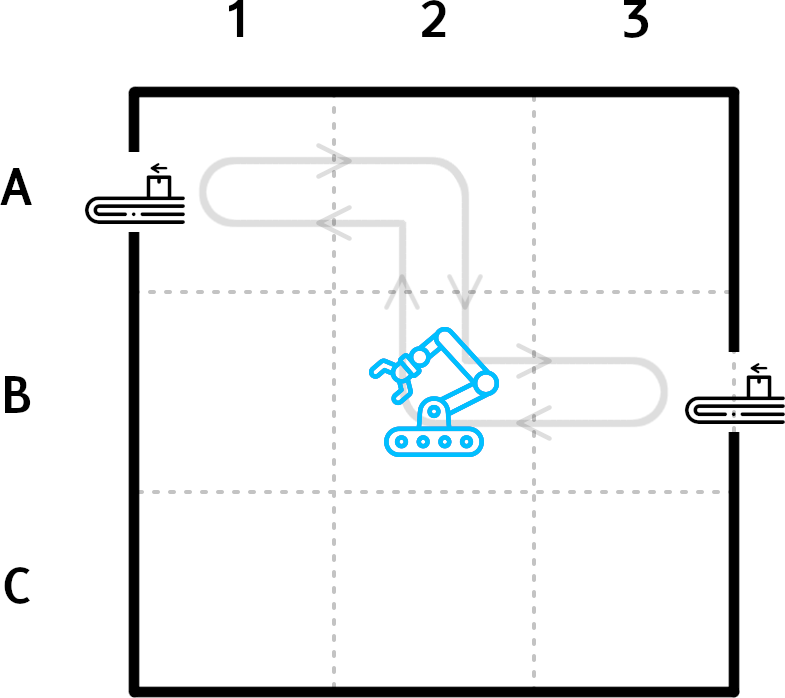}
	\hfill
	\includegraphics[width=0.45\columnwidth]{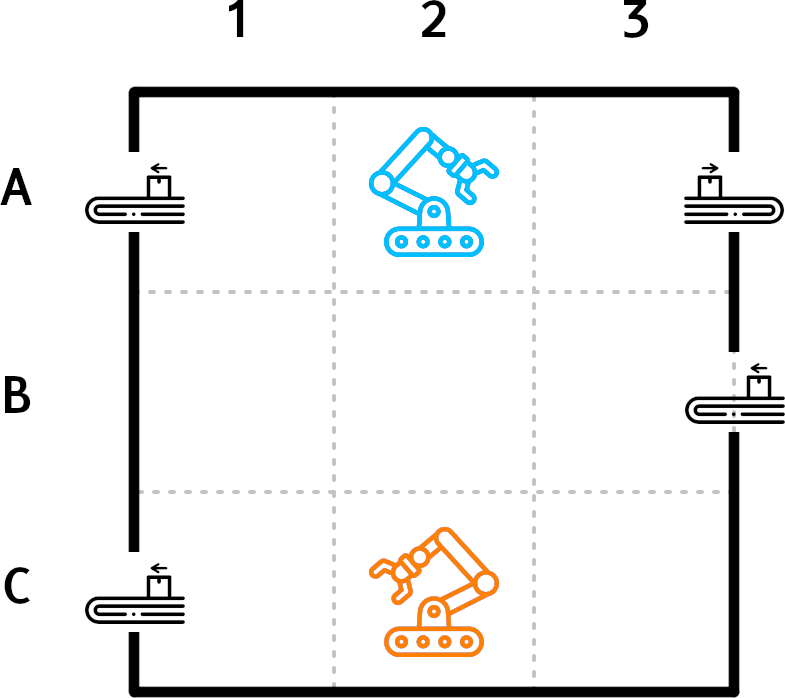}
	\caption{Illustration of a factory with two mobile robots $\ro$ and $\rt$, discussed in \cref{exp:pens:1}-\ref{exp:pens:3}. Cell $\Gamma_i$ is located in line $\Gamma\in\{A,B,C\}$ and row $i\in\{1,2,3\}$. Walls are indicated by solid lines, conveyor belts are depicted schematically.}\label{fig:factory}%
	\vspace{-0.3cm}
\end{figure}

The resulting controller synthesis problem can be modeled as a game over a finite graph $G$ -- the vertices remember the current position of each robot and its edges model all possible movements between them. Each player in the game models one robot and chooses the moves of that robot along the edges of the graph. Then the specifications for each robot are an LTL objective over $G$, formally resulting in a \emph{two-objective parity}\footnote{Parity games are formally introduced in \cref{sec:prelim}. Their expressivity is needed to allow for the full class of LTL specifications.} \emph{game}. While we exclude continuous robot dynamics from the subsequent discussion, we note that non-trivial dynamics can be abstracted into a game graph using well established abstraction techniques (see e.g.,\cite{tabuada2009verification,alur2015principles,belta2017formal} for an overview) or handled by a well-designed hand-over mechanism between continuous and logical feedback controllers (see e.g. \cite{ContextTriggerdABCD_2023}).
 \end{example}

Technically, the main contribution of this paper is a new algorithm to solve such \emph{two-objective parity games} arising from distributed logical control problems (as outlined in \cref{exp:pens:1}) in a \emph{distributed fashion}, i.e., without sharing the local specifications of components with each other, and by performing most computations locally. \emph{Assume-guarantee contracts} have proven to be very useful for such distributed synthesis problems and have been applied to various variants thereof \cite{chatterjee2007assume,BrenguierRaskinSankur2017,DBLP:journals/iandc/AlurMT18,DBLP:conf/tacas/BloemCJK15,
Bosy-based-contract,
Benveniste,Compositional_PerezGirard_2021,DallalTabuada_2015,8264194,8462741,9993344}. The main differences of our work compared to these works is threefold:\\
\begin{inparaenum}[(i)]
 \item We focus on \emph{logical} control design but thereby for the \emph{full class} of $\omega$-regular specifications, making our paper most related to other A/G-based distributed \emph{reactive synthesis} approaches \cite{Bosy-based-contract,AGDistributed,chatterjee2007assume,BrenguierRaskinSankur2017}. \\
 \item We \emph{co-design} both \emph{contracts and controllers}, i.e., we do not assume contracts to be given, as e.g.\ in \cite{Benveniste,9993344}. \\
 \item We develop a \emph{sound and complete} -- yet distributed -- synthesis framework. Existing tools either perform synthesis centrally, e.g.\ \cite{chatterjee2007assume,Bosy-based-contract,DBLP:conf/tacas/BloemCJK15}, are not complete, e.g.\ \cite{DallalTabuada_2015,AGDistributed,DBLP:journals/iandc/AlurMT18}, or cannot handle the full class of LTL specifications \cite{DallalTabuada_2015,Compositional_PerezGirard_2021}. %
\end{inparaenum}

Furthermore, the existing centralized synthesis techniques output a single strategy profile — one strategy per subsystem — which heavily depend on each other. The main advantage of our approach is the computation of decoupled strategy templates. This gives each subsystem the flexibility to independently choose any strategy from a huge class of strategies.

We achieve these new features in A/G-based synthesis by utilizing the concept of \emph{permissive templates} recently introduced by Anand et al. \cite{SIMAssumptions22,SIMController22}. Such \emph{templates} collect an infinite number of controller candidates in a concise data structure and allow for very efficient computability, adaptability and compositionality as illustrated with the next example.

 \begin{example}\label{exp:pens:2}
 Consider again  $ \ro $ in \cref{fig:factory} (left) with specification $\varphi_1$ from \cref{exp:pens:1}. A classical reactive synthesis engine would return a single strategy, e.g.\ one which keeps cycling along the path\footnote{We do not mention the not-so-relevant parts of the strategy for ease of understanding.} $ A_1\rightarrow A_2\rightarrow B_2\rightarrow B_3\rightarrow B_2\rightarrow A_2\rightarrow A_1$. However, $ \ro $ does not really need to stick to a single path to fulfill $\varphi_1$. It only needs to \emph{always eventually} go from $ A_1 $ to $ A_2 $ or $ B_1 $, from $ A_2 $ and $ B_1 $ to $ A_3 $ or $ B_2 $, from $ A_3 $ and $ B_2 $ to $ B_3 $, and so on. These very local liveness properties capture the essence of every correct controller for $\ro$ and can be summarized as a 
 \emph{strategy template}, which can be extracted from a classical synthesis engine without computational overhead \cite{SIMController22}. This controller representation has various advantages. 
 
 First, strategy templates are composable. If $ \ro $ is suddenly required to also collect and deposit goods from $A_3$, we can independently synthesize a strategy template for objective $\varphi_1'=\square\lozenge A_3 $ and both templates can be \emph{composed} by a simple conjunction of all present liveness properties\footnote{Of course, this is not always that easy, but provably so in most practical applications. For details see the extensive evaluation in \cite{SIMController22}.}. $\ro$ can choose a strategy that satisfies both objectives by complying with all template properties.
 
 Second, strategy templates keep all possible strategy choices in-tact and hence allow for a robust control implementation. If, e.g.\ due to the presence of other robots, $ A_1 $ is momentarily blocked, $\ro$ can keep visiting $ B_3 $ and $ A_3 $ until $ A_1 $ becomes available again\footnote{Again, see \cite{SIMController22} for an extensive case-study of this template feature.}. 
 
 Third, it was recently shown by Nayak et al.\ \cite{ContextTriggerdABCD_2023} that the flexibility of strategy templates allows to realize logical strategy choices by continuous feedback controllers over non-linear dynamics in a provable correct way, without time and space discretizations. 
 \end{example}

While the above example illustrates the flexibility of strategy templates for a single component, Anand et al. \cite{SIMAssumptions22,SIMController22} consider games with only one system (and objective). The current paper novelty leverages the easy compositionality and adaptability of the templates for \emph{contract-based distributed synthesis}. Intuitively, strategy templates collect the essence of strategic requirements for one robot in a concise data structure, which can be used to instantiate a contract. As these contracts are locally computed (one for each robot, w.r.t.\ its objective), they need to be synchronized, which might cause conflicts that need to be resolved. This requires multiple negotiation rounds until a realizing contract is found. The negotiation framework we present in this paper is ensured to \emph{always terminate} and to be sound and complete, i.e., to always provide a realizable contract if the synthesis problem has a solution. The intuition behind our framework is illustrated next.

\begin{example}\label{exp:pens:3}
Consider the two-robot scenario in \cref{fig:factory}(right) and observe that robot $ \ro $ has no strategy to satisfy its specification without any assumption on $\rt$'s behavior, e.g.\ if $ \rt $ always stays in $ B_3 $, $ \ro $ can never take raw material from $B_3$. %
However, since both robots are built by the factory designers, they can be designed to cooperate \enquote{just enough} to fulfill both objectives. We therefore assume that $ \ro $ can \enquote{ask} $ \rt $ 
to always eventually leave $ B_3 $, so $ \ro $ can collect the raw goods.
Algorithmically, this is done by locally computing both a strategy template \cite{SIMController22} for $\ro$ and an \emph{assumption template} \cite{SIMAssumptions22} on $\rt$. The latter contains local requirements on $\rt$ that need to be satisfied in order for $\ro$ to fulfill its objective.

When we then switch perspectives and (locally) synthesize a strategy for $\rt$ for its own objective $\varphi_2$, we can force $\rt$ to obey the assumption template from $\ro$'s previous computation. This, however, might again put new assumptions on $\ro$ for the next round. Due to the easy composability of templates, computations in each round are efficient, leading to an overall polynomial algorithm in the size of the game graph. 
In addition, the algorithm outputs a compatible pair of strategy templates (one for each player), which gives each player maximal freedom for realization (as outlined in \cref{exp:pens:2}) -- any control realization they pick allows also the other component to fulfill their objective. In this running example, the resulting pair of compatible strategy templates require $ \ro $ to go to $ A_1 $ when $ \rt $ is at $ B_3 $, and to $ B_3 $ only when the access is granted by $ \rt $ (which $ \rt $ will grant, by following its final template), and will also let $ \rt $ go to $ B_3 $ always eventually when it arrives.

As these resulting strategy templates only capture the \emph{essence} of the required cooperation, they can also be implemented under partial observation, as long as the required information about the other component are extractable from these observations. Due to page constrains, we omit the formal treatment of this case.

\end{example}

\noindent\textbf{Outline.}
After presenting required preliminaries in \cref{sec:prelim}, we formalize the considered contract-based synthesis problem in \cref{sec:overview}, and instantiate it via permissive templates in \cref{section:templates}. We then use this instantiation to devise a negotiation algorithm for its solution in \cref{sec:negotiate} and prove all features of the algorithm and its output illustrated in \cref{exp:pens:3}. 
Finally, \cref{sec:negotiate} provides empirical evidence that our negotiation framework possess desirable computational properties by comparing our C++-based prototype tool \toolname to state-of-the-art solvers on a benchmark suite. %

\section{Preliminaries}\label{sec:prelim}
\smallskip\noindent\textbf{Notation.}
We use $\mathbb{N}$ to denote the set of natural numbers including zero.
Given two natural numbers $a,b\in\mathbb{N}$ with $a<b$, we use $[a;b]$ to denote the set $\set{n\in\mathbb{N} \mid a\leq n\leq b}$.
Let $\Sigma$ be a finite alphabet.
The notations $\Sigma^*$ and $\Sigma^\omega$ respectively denote the set of finite and infinite words over $\Sigma$. 
Given two words $u\in \Sigma^*$ and $v\in \Sigma^*\cup \Sigma^{\omega} $, the concatenation of $u$ and $v$ is written as the word $uv$. 

\smallskip\noindent\textbf{Game Graphs.}
A \emph{game graph} is a tuple $\gamegraph= \tup{V=V^0\cupdot V^1,E}$ where $(V,E)$ is a finite, directed graph with \emph{vertices} $ V $ and \emph{edges} $ E $, and 
$ \vertexz, \vertexo\subseteq V$ form a partition of $V$. W.l.o.g.\ we assume that for every $v\in V$ there exists $v'\in V$ s.t.\ $(v,v')\in E$. 
A \emph{play} originating at a vertex $v_0$ is an infinite sequence of vertices $\play=v_0v_1\ldots \in V^\omega$.

\smallskip\noindent\textbf{Winning Conditions.}
Given a game graph $\gamegraph$, a \emph{winning condition} (or \emph{objective}) is a set of %
plays %
specified using a formula $ \spec $ in \emph{linear temporal logic} (LTL) over the vertex set $V$, i.e., LTL formulas whose atomic propositions are sets of vertices from $V$. 
Then the set of desired infinite plays is given by the $\omega$-regular language $\lang(\gamegraph,\spec)\subseteq V^\omega$. 
When $G$ is clear from the context, we simply write $\lang(\spec)$.
The standard definitions of $\omega$-regular languages and LTL are omitted for brevity and can be found in standard textbooks \cite{baier2008principles}. 

A \emph{parity objective} $\spec = \paritygame(\priority)$ is given by the LTL formula
\begin{equation}\label{equ:parity}
	\paritygame(\priority) \coloneqq \bigwedge_{i\inodd [0;d]} \left(\square\lozenge \priorityset{i} \implies \bigvee_{j\ineven [i+1;d]} \square\lozenge \priorityset{j}\right),
\end{equation}
with \emph{priority set} $\priorityset{j} = \{v : \priority(v)=j\}$ for $0\leq j\leq d$ of vertices for some \emph{priority function} $\priority : V \rightarrow [0;d]$ that assigns each vertex a \emph{priority}. $\lang(\paritygame(\priority))$ contains all plays $\play$ for which the highest priority appearing infinitely often along $\play$ is even.
We note that every game with an arbitrary $\omega$-regular set of desired 
plays can be reduced to a parity game (possibly with a larger set of vertices) by standard methods \cite{baier2008principles}.

\smallskip\noindent\textbf{Games.} 
A \emph{two-player (turn-based) game} is a tuple $\game=\tup{\gamegraph,\spec}$, where $G$ is a game graph, and $ \spec $ is the \emph{winning condition} over $\gamegraph$. 
A \emph{two-player (turn-based) two-objective game} is a triple $\game=\tup{\gamegraph,\specz,\speco}$, where $G $ is a game graph, and $ \specz $ and $ \speco $ are \emph{winning conditions} over $\gamegraph$, respectively, for $ \pz $ and $ \po $. We call a $\game$ a parity game if all involved winning conditions are parity objectives.

\smallskip\noindent\textbf{Strategies.}
A strategy of $\p{i}$~(for $i\in\{0,1\}$) is a function $\strati\colon \vertex^*\vertexi\to \vertex$ such that for every $\play v \in \vertex^*\vertexi$ holds that $\strati(\play v)\in \edge(v)$. 
A \emph{strategy profile} $(\stratz,\strato)$ is a pair where $\strati$ is a strategy for $\p{i}$.
Given a strategy $\strati$, we say that the play $\play=v_0v_1\ldots$ is \emph{compliant} with $\strati$ if $v_{k-1}\in \vertexi$ implies $v_{k} = \strati(v_0\ldots v_{k-1})$ for all $k$.
We refer to a play compliant with $\strati$ and a play compliant with a strategy profile ($\stratz$,$\strato$) as a \emph{$ \strati $-play} and a \emph{$ \stratz\strato $-play}, respectively.

\smallskip\noindent\textbf{Winning.}
Given a game $\game=(\gamegraph,\spec)$, a play $ \play $ in $ \game $ is \emph{winning} if it satisfies%
\footnote{Throughout the paper, we use the terms \enquote{winning for objective $\spec$} and \enquote{satisfying $\spec$} interchangeably.}
$\spec$, i.e., $ \play\in\lang(\spec) $. 
A strategy $\strati$ for $\p{i}$ is \emph{winning from a vertex} $ v\in V $ if all $\strati$-plays from $ v $ are winning. 
A vertex $v\in V$ is  \emph{winning for $\p{i}$}, if there exists a $\p{i}$ winning strategy $\strati$ from $ v $. 
We collect all winning vertices of $\p{i}$ in the \emph{$\p{i}$ winning region} $\team{i}\spec\subseteq V$.
We say a $\p{i}$ strategy is \emph{winning} for $\p{i}$ if it is winning from every vertex in $\team{i}\spec$.

Furthermore, given a game $(\gamegraph,\spec)$, we say a strategy profile $(\stratz,\strato)$ is \emph{winning from a vertex} $ v\in V $ if the $\stratz\strato$-play from $ v $ is winning.
We say a vertex $v\in V$ is \emph{cooperatively winning}, if there exists a winning strategy profile $(\stratz,\strato)$ from $ v $.
We collect all such vertices in the \emph{cooperative winning region} $\team{0,1}\spec\subseteq V$.
We say a strategy profile is \emph{winning} if it is \emph{winning} from every vertex in $\team{0,1}\spec$.
Winning strategies and cooperative winning region for a two-objective game $(\gamegraph,\spec_0,\spec_1)$ are defined analogously. %

\section{Contract-Based Synthesis}\label{sec:overview}
Towards a formalization of our proposed negotiation framework for distributed synthesis %
this section introduces the notion of assume-guarantee contracts (\cref{sec:contracts}) that we build upon, the notion of \irmac-specifications (\cref{sec:irmac-specification}) that describes our main goal, and formally states the synthesis problem we solve in this paper (\cref{sec:statement}).

\subsection{Assume-Guarantee Contracts}\label{sec:contracts}
Given a two-objective game $\game= (\gamegraph, \specz, \speco)$ we define an \emph{assume-guarantee contract} over $\game$ --- a \emph{contract} for short --- as a tuple $\Contract:=(\Tuple{\AssumptionI{0}, \GuaranteeI{0}},\Tuple{\AssumptionI{1}, \GuaranteeI{1}})$
where $ \AssumptionI{i}$ and $ \GuaranteeI{i}$ are LTL specifications over the graph $\gamegraph$ called the \emph{assumption} and the \emph{guarantee} for player $i$, respectively. 
It is well known that such contracts provide a \emph{certified interface} between both players, if they are 
\begin{enumerate}[(i)]
 \item \emph{compatible}, i.e.,
 \begin{equation}\label{equ:contract:compatible}
\lang(\GuaranteeI{i})\subseteq \lang(\AssumptionI{1-i}),~\text{and}
\end{equation}
\item \emph{realizable} by both players from at least one vertex, i.e., 
\begin{equation}\label{equ:contract:realizable}
 \exists v\in V~.~\forall i\in\{0,1\}~.~v\in \team{i}(\AssumptionI{i}\Rightarrow(\GuaranteeI{i}\wedge\speci)).
\end{equation}
\end{enumerate}

Unfortunately, it is also well known that for the full class of $\omega$-regular contracts, conditions \eqref{equ:contract:compatible}-\eqref{equ:contract:realizable} are not strong enough to provide a sound (and complete) proof rule for \emph{verification}, let alone the harder problem of \emph{synthesis}. In verification, one typically resorts to strengthening the contracts with less expressive properties \cite{AG_Pnueli,AG_Stark,AG_Misra_Chandy,AG_Abadi_Lamport}. This approach was also followed by \cite{AGDistributed} for synthesis, requiring contracts to be safety formulas. This, however, always results in an unavoidable conservatism, resulting in incompleteness of the proposed approaches. 

Within this paper, we take a novel approach to this problem which does not restrict the expressiveness of the formulas in $\Tuple{\AssumptionI{i}, \GuaranteeI{i}}$ but rather liberally changes the considered local specification in \eqref{equ:contract:realizable} to one which is \enquote{well-behaved} for contract-based distributed synthesis. We then show, that this liberty does still result in a sound and complete distributed synthesis technique by developing an algorithm to compute such \enquote{well-behaved} specifications whenever the original two-objective game has a cooperative solution. Before formalizing this problem statement in \cref{sec:statement} we first define such \enquote{well-behaved} specifications, called \irmac -- \textsf{i}ndependently \textsf{R}ealizable and \textsf{ma}ximally {\textsf{C}ooperative.

\subsection{\irmac-Specifications}\label{sec:irmac-specification}
We begin by formalizing a new \enquote{well-behaved} local specification for contract realizability.

\begin{definition}[\textsf{iR}-Contracts]
A contract $\Contract:=(\Tuple{\AssumptionI{0}, \GuaranteeI{0}},\Tuple{\AssumptionI{1}, \GuaranteeI{1}})$ over a two-objective game $\game= (\gamegraph, \specz, \speco)$ is called \emph{independently realizable} (\textsf{iR}) from a vertex $v$ if \eqref{equ:contract:compatible} holds and for all $i\in\{0,1\}$
\begin{equation}\label{equ:contract:realizable:b}
    v\in\team{i}\specip~\text{with}~\specip:= \GuaranteeI{i}\wedge (\AssumptionI{i}\Rightarrow \speci),
\end{equation}
 where $\specip$ is called a \emph{contracted local specification}.
\end{definition}

Intuitively, \eqref{equ:contract:realizable:b} requires the guarantees to be realizable by \p{i} without the \enquote{help} of player \p{1-i}, i.e., unconditioned from the assumption, which is in contrast to \eqref{equ:contract:realizable}. It is therefore not surprising that \textsf{iR}-Contracts allow to solve the local contracted games $(\gamegraph,\specip)$ fully independently (and in a zero-sum\footnote{A zero-sum game is a two-player game where the opponent has the negated specification of the protagonist, i.e., $(G,\spec,\neg\spec)$, i.e., the opponent acts fully adversarially. As defined in \cref{sec:prelim}, this (standard version of) games is denoted by tuples $(G,\spec)$.} fashion) while still ensuring that the resulting strategy profile solves the original game $\game$.

\begin{proposition}\label{prop:ir}
Given a two-objective game $\game= (\gamegraph, \specz, \speco)$ with \textsf{iR}-contract $\Contract:=(\Tuple{\AssumptionI{0}, \GuaranteeI{0}},\Tuple{\AssumptionI{1}, \GuaranteeI{1}})$ realizable from a vertex $v$, and contracted local specifications $(\speczp,\specop)$, let $\strati$ be a winning strategy in the (zero-sum) game $(\gamegraph, \specip)$. %

Then the tuple $\Tuple{\stratz,\strato}$ is a winning strategy profile for $\game$ from~$v$.
\end{proposition}

\begin{proof}
As $v\in \team{i}\specip$, $\strati$ is winning from $v$ for game $(\gamegraph,\specip)$, every $\strati$-play from $v$ satisfies $\specip = \GuaranteeI{i}\wedge (\AssumptionI{i}\Rightarrow\speci)$.
Therefore, every $(\stratz,\strato)$-play from $v$ satisfies both $\speczp$ and $\specop$.
Now, let us show that $\lang(\speczp\cap\specop) \subseteq \lang(\specz\wedge\speco)$.
Using the definition of contracted local specifications and by \eqref{equ:contract:compatible}, we have
$ \lang(\speczp\wedge\specop)= \lang(\GuaranteeI{0}\wedge (\AssumptionI{0}\Rightarrow\specz)) \cap \lang(\GuaranteeI{1}\wedge (\AssumptionI{1}\Rightarrow\speco))\subseteq_{\eqref{equ:contract:compatible}} \lang(\AssumptionI{1}\wedge (\AssumptionI{0}\Rightarrow\specz))\cap \lang(\AssumptionI{0}\wedge (\AssumptionI{1}\Rightarrow\speco))=\lang(\AssumptionI{0})\cap\lang(\AssumptionI{1})\cap\lang(\AssumptionI{0}\Rightarrow\specz)\cap\lang(\AssumptionI{1}\Rightarrow\speco)=\lang(\specz)\cap\lang(\speco) = \lang(\specz\wedge\speco)$.
Therefore, every $(\stratz,\strato)$-play from $v$ satisfies $\specz\wedge\speco$. Hence, $(\stratz,\strato)$ is a winning strategy profile for $\game$ from $v$.
\end{proof}

By the way they are defined, \textsf{iR}-Contracts can be used to simply encode a single winning strategy profile from a vertex, which essentially degrades contract-based synthesis to solving a single cooperative game with specification $\specz\cup\speco$. The true potential of \textsf{iR}-Contracts is only reveled if they are reduced to the \enquote{essential cooperation} between both players. Then the local contracted specifications $\specip$ will give each player as much freedom as possible to choose its local strategy. This is formalized next.

\begin{definition}[\textsf{iRmaC}-Specifications]
 Given a two-objective game $\game= (\gamegraph, \specz, \speco)$, a pair of specifications $(\speczp,\specop)$ is said to be \emph{independently realizable} and \emph{maximally cooperative} (\textsf{iRmaC}) if
 \begin{subequations}\label{equ:ceis}
 \begin{align}
\lang(\specz\wedge\speco) &= \lang(\speczp\wedge\specop), \text{ and}\label{equ:ceis:ce}\\
 \team{0,1} (\specz\wedge\speco) &= \team{0}\speczp~\cap~\team{1}\specop.\label{equ:ceis:is}
\end{align}
\end{subequations}
\end{definition}
\noindent Here  \eqref{equ:ceis:ce} ensures that the contracted local games $(\gamegraph,\specip)$ do not eliminate any cooperative winning play allowed by the original specifications, while \eqref{equ:ceis:is} ensures that the combination of local winning regions does not restrict the cooperative winning region. These properties of \irmac-specifications now allow each player to extract a strategy $\strati$ \emph{locally and fully independently} by solving the (zero-sum) game $(\gamegraph,\specip)$. Then it is guaranteed that the resulting (independently chosen) strategy profile $(\stratz,\strato)$ is winning in $\game= (\gamegraph, \specz, \speco)$. This is formalized next.

\begin{proposition}\label{prop:irmac}
Given a two-objective game $\game= (\gamegraph, \specz, \speco)$ with \textsf{iRmaC} specifications-$(\speczp,\specop)$, the following are equivalent: %
\begin{enumerate}[(i)]
    \item there exists a winning strategy profile from $v$ in $(\gamegraph,\specz,\speco)$,\label{item:prop:irmac1}
    \item for each $i\in\{0,1\}$, there exists a $\p{i}$ winning strategy from $v$ in $(\gamegraph,\specip)$.\label{item:prop:irmac2}
\end{enumerate}
\begin{proof}
    \begin{inparaitem}
        \item[(\cref{item:prop:irmac1} $\Rightarrow$\cref{item:prop:irmac2})] If there exists a winning strategy profile from $v$ for the game $(\gamegraph,\specz,\speco)$, then $v\in\team{0,1}(\specz\wedge\speco)$. Then, by \eqref{equ:ceis:is}, $v\in\team{i}\specip$ for each $i\in\{0,1\}$. Hence, there exists a $\p{i}$ winning strategy from $v$ in $(\gamegraph,\specip)$ for each $i\in\{0,1\}$.
        \item[(\cref{item:prop:irmac1}$\Leftarrow$\cref{item:prop:irmac2})] Similarly, if there exists a $\p{i}$ winning strategy from $v$ for the game $(\gamegraph,\specip)$ for each $i\in\{0,1\}$, then $v\in\team{i}\specip$ for each $i\in\{0,1\}$. Then, by \eqref{equ:ceis:is}, $v\in\team{0,1}(\specz\wedge\speco)$, and hence, there exists a winning strategy profile from $v$ for the game $(\gamegraph,\specz,\speco)$.
    \end{inparaitem} 
\end{proof}
\end{proposition}

With this, we argue that \irmac-specifications indeed provide a  \emph{maximally cooperative} contract for distributed synthesis which allows to fully decentralize remaining strategy choices.

\subsection{Problem Statement and Outline}\label{sec:statement}
Based on the desirable properties of \irmac-specifications outlined before, the main contribution of this paper is an algorithm to compute \irmac-specifications for \emph{two-objective parity games}, which are a canonical representation of two-player games with an LTL objective for each player. %

\begin{problem}\label{prob:main}
 Given a two-objective parity game $\game= (\gamegraph, \specz, \speco)$, compute \irmac-specifications $(\speczp,\specop)$.
\end{problem}

In particular, we provide an algorithm which \emph{always} outputs an \irmac-specification s.t.\ the latter only results in an empty cooperative winning region (via \eqref{equ:ceis:is}) \emph{if and only if} $\game$ is not cooperatively solvable. Thereby, our approach constitutes a \emph{sound and complete} approach to distributed logical controller synthesis. All existing solutions to this problem, i.e., \cite{Bosy-based-contract,AGDistributed}, only provide a sound approach. In addition, as outlined before, the computed \irmac-specifications then allow to choose winning strategies (i.e., controllers) in a fully decentralized manner (due to \cref{prop:irmac}). 

Our algorithm for solving \cref{prob:main} is introduced in \cref{section:templates} and \cref{sec:negotiate}. Conceptually, this algorithm builds upon the recently introduced formalism of \emph{permissive templates} by Anand et al. \cite{SIMAssumptions22,SIMController22} and utilizes their efficient computability, adaptability and permissiveness to solve \cref{prob:main}. Interestingly, this  approach does not only allow us to solve \cref{prob:main} but also allows resulting local strategies to be easily adaptable to new local objectives and unforeseen circumstances, as illustrated in the motivating example of \cref{sec:intro}. We showcase the computational efficiency and the extra features of our approach by experiments with a prototype implementation on a set of control-inspired benchmarks in \cref{sec:experiements}.

\section{Contracts as Templates}\label{section:templates}
This section shows how \emph{templates} can be used to solve \cref{prob:main} and starts with an illustrative example to convey some intuition.

\begin{figure*}
\vspace{-0.5cm}
\centering
 \begin{tikzpicture}
      \node[player0] (0) at (0, 0) {$a$};
      \node[player1] (1) at (\fpeval{\pos}, 0) {$b$};
			\node[player1] (2) at (-\fpeval{\pos},0) {$c$};
			\node[player0] (3) at (2*\fpeval{\pos}, 0) {$d$};
			\path[->] (0) edge[loop above] () edge[bend left = 0] (1) edge[bend left = 20] (2);
			\path[->] (1) edge[loop above] () edge[bend left = 0] (3);
			\path[->] (2) edge[loop above] () edge[bend left = 20] (0);
			\path[->] (3) edge[bend left = 20] (2.south east);

\node[text width=8cm,anchor=north west] at (4,0.6) {
$\spec_0=\Box\Diamond\{c\}$\\
$\spec_1=\Diamond\Box\{a,c,d\}$ \\
$\spec_1'=\Diamond\Box\{a,b,c\}$
}; 

\node[text width=8cm,anchor=north west] at (6.8,0.6) {
$\Rightarrow$ $\assumpI{0}=\templategrlive(\{e_{bd}\})$,\\
$\Rightarrow$ $\assumpI{1}=\templatecolive(e_{ab})$,\\
$\Rightarrow$ $\assumpI{1}'=\mathtt{true}$ ,
};

\node[text width=8cm,anchor=north west] at (10.2,0.6) {
$\StratI{0}=\templategrlive(\{e_{ac}\})$\\
$\StratI{1}=\templatecolive(e_{bb})$\\
$\StratI{1}'=\templatecolive(e_{bd})$
};
\end{tikzpicture}
\vspace{-0.8cm}
 \caption{A two-player game graph discussed in \cref{sec:example} with $\po$ (squares) and $\pz$ (circles) vertices, different winning conditions $\spec_i$, and corresponding winning assumption templates $\assump_i$ and strategy templates $\Strat_i$ for Player $i$.}\label{fig:introExample}
\end{figure*}

\begin{example}\label{sec:example}
In order to appreciate the simplicity, adaptability and compositionality of templates consider the two-objective game in \cref{fig:introExample}. 
The winning condition $\specI{0}$ for $\p{0}$ requires vertex $c$ to be seen infinitely often. Intuitively, every winning strategy for \p{0} w.r.t.\ $\specI{0}$ needs to eventually take the edge~$e_{ac}$ if it sees vertex $a$ infinitely often. 
Furthermore, \p{0} can only win from vertex $b$ with the help of \p{1}. In particular, \p{1} needs to ensure that  whenever vertex~$b$ is seen infinitely often it takes edge~$e_{bc}$ infinitely often. These two conditions can be concisely formulated via the strategy template $\Stratz=\templategrlive(\{e_{ac}\})$ and an assumption template $\assumpI{0}=\templategrlive(\{e_{bd}\})$, both given by what we call a live-edge template -- if the source is seen infinitely often, the given edge has to be taken infinitely often. 
 It is easy to see that every $\p{0}$ strategy that satisfies $\Stratz$ is winning for $\specI{0}$ under the assumption that $\p{1}$ chooses a strategy that satisfies $\assumpz$.
 
Now, consider the winning condition $\speco$ for $\p{1}$ which requires the play to eventually stay in region $\{a,c,d\}$. This induces assumption $\assumpo$ on $\p{0}$ and strategy template $\Strato$ for $\p{1}$ given in \cref{fig:introExample} (right). Both are \emph{co-liveness} templates -- the corresponding edge can only be taken \emph{finitely} often. This ensures that all edges that lead to the region $\{a,c,d\}$ are taken only finitely often.

The tuples of strategy and assumption templates $(\assump_i,\Strati)$  we have constructed for both players in the above example will be called \emph{contracted strategy-masks}, \csm for short. 
If the players now share the assumptions from their local \csms, it is easy to see that in the above example both players can ensure the assumptions made by other player in addition to their own strategy templates, i.e., each $\p{i}$ can realize $\assumpI{1-i}\wedge\Strati$ from all vertices. In this case, we call the \csms $(\assump_i,\Strati)$ \emph{compatible}. 
In such situations, the new specifications $(\speczp,\specop)$ with $\specip = \assumpI{1-i}\wedge (\assumpi\Rightarrow\speci)$ are directly computable from the given \csms and indeed form an \irmac-contract.

Unfortunately, locally computed \csms are not always compatible. To see this, consider the slightly modified winning condition $\specI{1}'$ for $\p{1}$ that induces strategy template $\Strato'$ for $\p{1}$. This template  requires the edge $e_{bd}$ to be taken only \emph{finitely} often. Now, $\p{1}$ cannot realize both $\assumpz$ and $\Strato'$ as the conditions given by both templates for edge $e_{bd}$ are \emph{conflicting} -- the same edge cannot be taken infinitely often \emph{and} finitely often.
In this case one more round of negotiation is needed to ensure that both players eventually avoid vertex $d$ by modifying the objectives to $\speci' = \speci\wedge\LTLeventually\LTLalways \neg d$. This will give us a new pair of \csms that are indeed compatible, and a new pair of objectives $(\speczp,\specop)$ that are now again an \irmac specification.
\end{example}

In the following we formalize the notion of templates (\cref{section:templates:recal}) and \csms (\cref{sec:csm}) and show that, if compatible, they indeed provide \irmac-specifications (\cref{sec:templates:solve}). We further show how to compute \csms for each player (\cref{section:template extraction}). The outlined negotiation for compatibility is then discussed in \cref{sec:negotiate}.

\subsection{Permissive Templates}\label{section:templates:recal}
This section recalls the concept of \emph{templates} from \cite{SIMAssumptions22,SIMController22}.
In principle, a template is simply an LTL formula $\template$ over a game graph $G$. We will, however, restrict attention to four distinct types of such formulas, and interpret them as a succinct way to represent a set of strategies for each player, in particular all strategies that \emph{follow} $\template$. Formally, a $ \p{i} $ strategy $\strati$ \emph{follows} $ \template $ if every $ \strati $-play belongs to $\lang(\template) $, i.e., strategy $\strati$ is winning from all vertices in the game $(\gamegraph,\template)$.
The exposition in this section follows the presentation in \cite{SIMAssumptions22} where more illustrative examples and intuitive explanations can be found.

\smallskip\noindent\textbf{Safety Templates.} Given a set $ \safegroup\subseteq E $ of \emph{unsafe edges}, the safety template is defined as
	$\templatesafe(\safegroup)\coloneqq \LTLalways \wedge_{e\in \safegroup} \neg e$,
where an edge $e=(u,v)$ is equivalent to the LTL formula $u\wedge\LTLnext v$. A safety template requires that an edge to $S$ should never be taken. %

\smallskip\noindent\textbf{Live-Group Templates.} A \emph{live-group} $\livegroupSingleN = \Set{e_j}_{j\geq 0}$ is a set of edges $e_j = (s_j,t_j)$ with source vertices $\src(\livegroupSingleN):=\Set{s_j}_{j\geq 0}$. Given a set of live-groups 
$\livegroup=\left\{\livegroupSingle\right\}_{i\geq 0}$  we define a live-group template as
	$\templategrlive(\livegroup) \coloneqq \bigwedge_{i\geq 0}\square\lozenge src(\livegroupSingle)\Rightarrow\square\lozenge \livegroupSingle$.
A live-group template requires that if some vertex from the source of a live-group is visited infinitely often, then some edge from this group should be taken infinitely often by the following strategy. 

\smallskip\noindent\textbf{Conditional Live-Group Templates.}  A \emph{conditional live-group} over $\gamegraph$ is a pair $ (R, \livegroup ) $, where $ R\subseteq V $ and $ \livegroup $ is a set of live groups. Given a set of conditional live groups $ \condlivegroup $ we define a \emph{conditional live-group template} as 
	$\templatecondlive(\condlivegroup) \coloneqq \bigwedge_{(R,\livegroup)\in \condlivegroup}\left(\square\lozenge R\Rightarrow \templategrlive(\livegroup)\right)$.
A conditional live-group template requires that for every pair $(R,\livegroup)$, if some vertex from the set $R$ is visited infinitely often, then a following strategy must follow the live-group template $\templategrlive(\livegroup)$.

\smallskip\noindent\textbf{Co-liveness Templates.} Given a set of \emph{co-live} edges $ \colivegroup $ a co-live template is defined as
	$\templatecolive(\colivegroup) \coloneqq \bigwedge_{e\in\colivegroup}\lozenge\square \neg e$.
A co-liveness template requires that edges in $ \colivegroup $ are only taken finitely often.

\smallskip\noindent\textbf{Composed Templates.}
In the following, a template $\template:=\templatesafe(\safegroup)\wedge\templatecolive(\colivegroup)\wedge\templatecondlive(\condlivegroup)$ will be associated with the tuple $(\safegroup,\colivegroup,\condlivegroup)$, denoted by $\template\compt(\safegroup,\colivegroup,\condlivegroup)$.
Similarly, $\template\compt(\safegroup,\colivegroup,\livegroup)$ denotes the template $\template:=\templatesafe(\safegroup)\wedge\templatecolive(\colivegroup)\wedge\templategrlive(\livegroup)$.
We further note that the conjunction of two templates $\template\compt(\safegroup,\colivegroup,\condlivegroup)$ and $\template'\compt(\safegroup',\colivegroup',\condlivegroup')$ is equivalent to the template $(\template\wedge\template')\compt(\safegroup\cup\safegroup',\colivegroup\cup\colivegroup',\condlivegroup\cup\condlivegroup')$ by the definition of conjunction of LTL formulas.

\subsection{Contracted Strategy-Masks (\csms)}\label{sec:csm}
Towards our goal of formalizing \irmac-specifications via templates, this section defines \emph{contracted strategy-masks} 
which contain two templates $\assumpi$ and $\Strati$, representing a set of \p{1-i}- and \p{i}-strategies respectively, which can be interpreted as the assumption $\assumpi$ on player \p{1-i} under which \p{i} can win the local game $(\gamegraph,\speci)$ with any strategy from $\Strati$.

Towards this goal, we first observe that every template in \cref{section:templates:recal} is defined via a set of edges that a following strategy needs to handle in a particular way. Intuitively, we can therefore \enquote{split} each template into a part restricting strategy choices for \p{0} (by only considering edges originating from $V_0$) and a part restricting strategy choices for \p{1} (by only considering edges originating from $V_1$), which then allows us to define \csm.

\begin{definition}
 Given a game graph $\gamegraph=(V,E)$, a template $\template\compt(\safegroup,\colivegroup,\condlivegroup)$ over $\gamegraph$ is an \emph{assumption template} (resp. a \emph{strategy template}) for player $i$ if for all edges $e\in\safegroup\cup\colivegroup\cup\overline{\livegroupSingleN}$ holds that $src(e)\in V_{1-i}$ (respectively $src(e)\in V_i$) where $\overline{\livegroupSingleN}:=\bigcup\{\livegroupSingleN\in\livegroup~\mid~(\cdot,\livegroup)\in\condlivegroup\}$.
\end{definition}

\begin{definition}
 Given a game $(\gamegraph,\speci)$, a \emph{contracted strategy-mask} (\csm) for player $i$ is a tuple $(\assumpI{i}, \StratI{i})$, such that $\assumpI{i}\compt(\safegroupS_i,\colivegroupS_i,\condlivegroupS_i)$ and $\StratI{i}\compt(\safegroupA_i, \colivegroupA_i, \condlivegroupA_i)$ are assumption and strategy templates for player $i$, respectively.
\end{definition}

We next formalize the intuition that \csms collect winning strategies for \p{i} under assumptions on \p{1-i}.

\begin{definition}\label{def:csm_winning}
A \csm $(\assumpI{i}, \StratI{i})$ is \emph{winning} for Player $i$ in $(\gamegraph,\speci)$ \emph{from vertex} $v$ if for every $ \p{i} $ strategy $ \strati $ following $ \StratI{i} $ and every $\p{1-i}$ strategy $\stratI{1-i}$ following $\assumpI{i}$ the $\stratz\strato$-play originating from $v$ is winning.
Moreover, we say a \csm $(\assumpI{i}, \StratI{i})$ is \emph{winning} for Player $i$ in $(\gamegraph,\speci)$ if it is winning from every vertex in $\team{0,1}\speci$.
\end{definition}
We denote by $\team{i}(\assumpI{i}, \StratI{i})$ the set of vertices from which $(\assumpI{i}, \StratI{i})$ is winning for Player $i$ in $(\gamegraph,\speci)$.
Due to localness of our templates, the next remark follows.
\begin{remark}\label{rem:winundera}
If a \csm $(\assumpI{i}, \StratI{i})$ is winning for Player $i$ in $(\gamegraph,\speci)$ from vertex $v$, then every $\p{i}$ strategy $\strati$ following $\Strati$ is winning for $\p{i}$ in the game $(\gamegraph,\assumpi\Rightarrow\speci)$.
\end{remark}

\subsection{Representing Contracts via \csm}\label{sec:templates:solve}
The previous subsection has formalized the concept of a \csm for a local (zero-sum) game $(\gamegraph,\speci)$ . This section now shows how under which conditions the combination of two \csms $(\assumpI{0}, \StratI{0})$ and $(\assumpI{1}, \StratI{1})$ (one for each player) allows to construct a contract
 \begin{equation}\label{equ:contract-csm}
  \Contract:=((\assumpI{0},\assumpI{1}),(\assumpI{1},\assumpI{0})),
 \end{equation}
(i.e, setting  $\AssumptionI{i}:= \assumpI{i}$ and $\GuaranteeI{i}:= \assumpI{1-i}$), which induces \irmac-specifications $(\speczp,\specop)$ as in \eqref{equ:contract:realizable:b}. %

The first condition we need is compatibility.
\begin{definition}[Compatible \csms]\label{def:compatible}
	Two \csms, $(\assumpI{0}, \StratI{0})$ for $\p{0}$ and $(\assumpI{1}, \StratI{1})$ for $\p{1}$, are said to be \emph{compatible}, if for each $i\in\{0,1\}$, there exists a $\p{i}$ strategy $\strati$ that follows $\Strati\wedge\assumpI{{1-i}}$.
\end{definition}
Intuitively, as $\assumpI{1-i}$ is the assumption on $\p{i}$ and $\Strati$ represents the template that $\p{i}$ will follow, we need to find a strategy that follows both templates. 
Before going further, let us first show a simple result that follows from \cref{def:compatible}.
\begin{proposition}\label{prop:csm-iR}
 Given a two-objective game $\game= (\gamegraph, \specz, \speco)$, let $(\assumpI{0}, \StratI{0})$ and $(\assumpI{1}, \StratI{1})$ be two compatible \csms s.t.\ $(\assumpI{i}, \StratI{i})$ is winning from a vertex $v$ for \p{i} in $(\gamegraph,\speci)$. Then the contract $\Contract$ as in \eqref{equ:contract-csm} is an \ir-contract realizable from $v$.
\end{proposition}
\begin{proof}
We need to show that $v\in\team{i}(\assumpI{1-i}\wedge(\assumpi\Rightarrow\speci))$ for each $i=0,1$.
Firstly, as the \csms are compatible, for each $i$, there exists a $\p{i}$ strategy $\strati$ that follows $\Strati\wedge\assumpI{1-i}$. 
Hence, every $\strati$-play satisfies both $\assumpI{1-i}$.
Secondly, as \csm $(\assumpi,\Strati)$ is winning from $v$ for $\p{i}$ in game $(\gamegraph,\speci)$, by \cref{rem:winundera}, every $\strati$-play from $v$ satisfies $\assumpI{i}\Rightarrow\speci$.
Therefore, every $\strati$-play from $v$ satisfies $\assumpI{1-i}\wedge(\assumpi\Rightarrow\speci)$, and hence, $v\in\team{i}(\assumpI{1-i}\wedge(\assumpi\Rightarrow\speci))$.
\end{proof}

To ensure that two compatible \csms as in \cref{prop:csm-iR} are not only an \ir-contract but also provide \irmac-specifications, we utilize the main result from \cite{SIMAssumptions22} which showed that assumption templates can be computed in an \emph{adequately permissive}\footnote{We refer to \cite{SIMAssumptions22} for an elaborate discussion of conditions (i)-(iii) in \cref{def:adequate assumption}.} way over a given parity game. This notion is translated to \csms next.

\begin{definition}\label{def:adequate assumption}
Given a game $(\gamegraph,\speci)$ and a \csm $(\assumpI{i}, \StratI{i})$ for Player $i$, we call this \csm \emph{adequately permissive} for $ (\gamegraph,\speci) $ if it is
	\begin{enumerate}[(i)]
		\item \emph{sufficient}: $\team{i}(\assumpI{i}, \StratI{i}) \supseteq \team{0,1}\speci$,\label{item:def:APA:sufficient}
		\item \emph{implementable}: $\team{1-i}\assump_i = V$ and $\team{i}\Strati = V$
		\item \emph{permissive}: $\lang(\speci)\subseteq \lang(\assumpi)$.\label{item:def:APA:permissive}
	\end{enumerate}
\end{definition}

Note that the sufficiency condition makes the \csm winning as formalized in the next remark.
\begin{remark}\label{rem:winningcsm}
	If a \csm $(\assumpi,\Strati)$ for $\p{i}$ in a game $(\gamegraph,\speci)$ is sufficient, then it is winning for $\p{i}$ in $(\gamegraph,\speci)$.
\end{remark}

With this, we are ready to prove the main theorem of this section, which shows that synthesis of \irmac-specifications reduces to finding adequately permissive \csms which are compatible.

\begin{theorem}\label{thm:csm-iRmaC}
Given a two-objective game $\game= (\gamegraph, \specz, \speco)$, let $(\assumpI{0}, \StratI{0})$ and $(\assumpI{1}, \StratI{1})$ be two \emph{compatible} \csms s.t.\ $(\assumpI{i}, \StratI{i})$ is \emph{adequately permissive} for \p{i} in $(\gamegraph,\speci)$. Then $(\speczp,\specop)$ with $\specip = \assumpI{1-i}\wedge(\assumpI{i}\Rightarrow\speci)$ are \textsf{iRmaC}-specifications.
\end{theorem}
\begin{proof}
	We need to show that the pair $(\speczp,\specop)$ satisfies  \eqref{equ:ceis:ce} and \eqref{equ:ceis:is}.
	The proof for \eqref{equ:ceis:ce} is completely set theoretic:\\
	\begin{inparaitem}
		\item[$ (\subseteq) $] For each $ i\in\{0,1\} $, it holds that $\lang(\speci)\subseteq_{\eqref{item:def:APA:permissive}} \lang(\speci)\cap\lang(\assumpI{i})\subseteq  \lang(\assumpI{i}\Rightarrow \speci)\cap\lang(\assumpI{i})$.
		With this we have, $\lang(\specz\wedge\speco)\subseteq \lang(\specz)\cap\lang(\speco)
			\subseteq  \lang(\assumpI{0}\Rightarrow \specz)\cap\lang(\assumpz) \cap  \lang(\assumpI{1}\Rightarrow \speco)\cap\lang(\assumpo)$ which simplifies to $\lang(\speczp)\cap \lang(\specop) = \lang(\speczp\wedge\specop)$.\\
		\item[$ (\supseteq) $] For each $ i\in\{0,1\} $, it holds that $\lang(\specip)$ is equivalent to
		$\lang(\assumpI{1-i}\wedge (\assumpI{i}\Rightarrow\speci))=\lang(\assumpI{1-i}\wedge (\neg \assumpI{i}\vee \speci))=\lang((\assumpI{1-i}\wedge\neg \assumpI{i})\vee (\assumpI{1-i}\wedge \speci))$ which simplifies to $\lang(\assumpI{1-i}\wedge\neg \assumpI{i})\cup\lang(\assumpI{1-i}\wedge \speci)$.
		Then we have that $\lang(\speczp\wedge\specop)=\lang(\speczp)\cap \lang(\specop)$ reduces to $\left(\lang(\assumpI{1}\wedge\neg \assumpI{0})\cup\lang(\assumpI{1}\wedge \specz)\right)\cap
				\left(\lang(\assumpI{0}\wedge\neg \assumpI{1})\cup\lang(\assumpI{0}\wedge \speco)\right)$ which simplifies to $\lang(\assumpI{1}\wedge \specI{0})\cap\lang(\assumpI{0}\wedge \speco)
			\subseteq \lang(\specz\wedge\speco)$.
	\end{inparaitem}
	
	Next, we show that one side of \eqref{equ:ceis:is} follows from \eqref{equ:ceis:ce}, whereas the other side follows from \cref{prop:csm-iR}: 
	\begin{inparaitem}
		\item[$ (\supseteq) $] If $v \in \team{0}\speczp~\cap~\team{1}\specop $, then for each $i\in\{0,1\}$, there exists a strategy $\strati$ for $\p{i}$ such that every $\strati$-play from $v$ belongs to $\lang(\specip)$. Hence, every $\stratz\strato$-play from $v$ belongs to $\lang(\speczp)\cap\lang(\specop) = \lang(\speczp\wedge\specop)=_{\eqref{equ:ceis:ce}}\lang(\specz\wedge\speco)$. Therefore, $v\in\team{0,1}(\specz\wedge\speco)$.
  		\item[$(\subseteq)$] If $v\in\team{0,1}(\specz\wedge\speco)\subseteq \team{0,1}\speci$, then by \cref{item:def:APA:sufficient}, $v\in \team{i}(\Strati,\assumpi)$. Hence, for each $i$, \csm $(\Strati,\assumpi)$ is winning for $\p{i}$ from $v$. As the \csms are also compatible, by \cref{prop:csm-iR}, the contract $\Contract = (\assumpz,\assumpo)$ is an \ir-contract realizable from $v$. Hence, by definition, $v\in \team{i}\assumpI{1-i}\wedge(\assumpI{i}\Rightarrow\speci) = \team{i}\specip$. Therefore, $v\in\team{0}\speczp \cap \team{1}\specop$.
	\end{inparaitem}
\end{proof}

\subsection{Computing Adequately Permissive \csms}\label{section:template extraction}
As stated before, \cref{thm:csm-iRmaC} shows that a solution to \cref{prob:main} reduces to finding \emph{adequately permissive} \csms which are \emph{compatible}. Due to the close connection between \emph{adequately permissive} \csms and \emph{adequately permissive assumption templates} from \cite{SIMAssumptions22}, it turns out that the computation of \emph{adequately permissive} \csms can be done in very close analogy to the computation of adequately permissive assumption templates for parity games from \cite{SIMAssumptions22}. In particular, we inherent (i) the observation that conjunctions of safety, co-live and conditional live-group templates are rich enough to express adequately permissive \csms, and (ii) the existence of a polynomial time (i.e., very efficient) algorithm for their construction.

\begin{theorem}\label{thm:parity_assumption_new}
	Given a game graph $ \gamegraph=\tup{V=\vertexz\cupdot \vertexo, E} $ with parity objective $ \spec_i$,
	an adequately permissive \csm $(\assump_i\compt(\safegroupS_i,\colivegroupS_i,\condlivegroupS_i), \Strat_i\compt(\safegroupA_i, \colivegroupA_i, \condlivegroupA_i))$  for player $i$ in $(\gamegraph,\spec_i)$ can be computed in time $ \bigO(n^4) $, where $n=|V|$. We call the respective procedure for this computation $\parityTemp(\gamegraph,\spec_i)$.
\end{theorem}

For completeness, we give the full algorithm for $\parityTemp$, along with its simplified (and more efficient) versions for Safety ($\computeSafe$ with $\bigO(m)$,~$m=|E|$), Büchi ($\buchiTemp$ with $\bigO(m)$,~$m=|E|$) and co-Büchi games ($\cobuchiTemp$ with $\bigO(m)$,~$m=|E|$) along with additional intuition and all correctness proofs in \cref{app:template extraction}. This exposition is given in very close analogy to \cite{SIMAssumptions22}.

With \cref{thm:parity_assumption_new} in place, the main algorithmic problem for solving \cref{prob:main} is to ensure that computed \csms are \emph{compatible}. This is done via a negotiation algorithm, as already illustrated in the last paragraph of \cref{sec:example}, formalized next.

\section{Negotiation for Compatible \csms}\label{sec:negotiate}
This section contains the main contribution of this paper w.r.t.\ the algorithmic solution of \cref{prob:main}. That is, we give an algorithm to compute \emph{adequately permissive} and \emph{compatible} \csms in a mostly distributed fashion. %

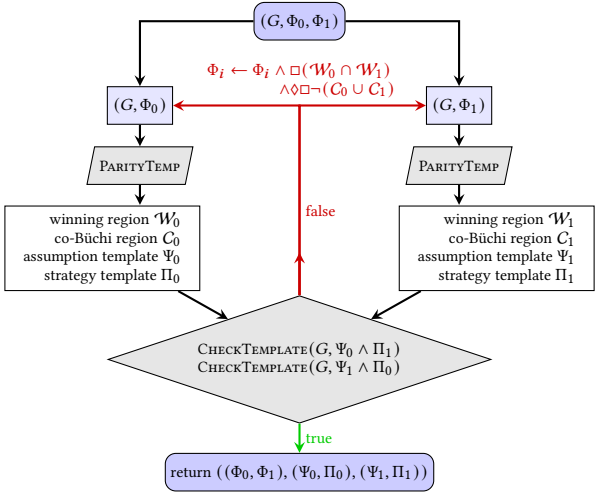
\begin{figure}
\scriptsize
\centering
\begin{tikzpicture}[node distance=3cm]

	\node (start) [startstop,fill=blue!20] {$(\gamegraph,\spec_0,\spec_1)$};
	\node (game0) [io, below left of=start, yshift=1cm,fill=blue!10] {$(\gamegraph,\spec_0)$};
	\node (game1) [io, below right of=start, yshift=1cm,fill=blue!10] {$(\gamegraph,\spec_1)$};
	\node (algo0) [process, below of=game0, yshift=1.3cm, yshift=0.9cm,fill=black!10] {\parityTemp};
	\node (algo1) [process, below of=game1, yshift=1.3cm, yshift=0.9cm,fill=black!10] {\parityTemp};
	\node (dtemp0) [below of=algo0, yshift=2.4cm]{};
	\node (temp0) [io, below of=algo0, xshift=-0.5cm, yshift=1.9cm] {
								\begin{tabular}{r} 
								\hspace*{0cm}winning region $\winz$\\
								\hspace*{0cm}\cobuchi region $\conflict_0$\\
								assumption template~$\assumpI{0}$\\
								\hspace*{0cm}strategy~template~$\StratI{0}$
								\end{tabular}};
	\node (dtemp1) [below of=algo1, yshift=2.4cm]{};
	\node (temp1) [io, below of=algo1, xshift=0.5cm, yshift=1.9cm] {
								\begin{tabular}{r} 
								\hspace*{0cm}winning region $\wino$\\
								\hspace*{0cm}\cobuchi region $\conflict_1$\\
								assumption template~$\assumpI{1}$\\
								\hspace*{0cm}strategy~template~$\StratI{1}$
								\end{tabular}};
	\node (dec1) [decision, below of=start, yshift=-1.5cm,fill=black!10] {
								\begin{tabular}{c}		
								$\checkTemplate(\gamegraph,\assumpI{0}\wedge\StratI{1})$\\
								$\checkTemplate(\gamegraph,\assumpI{1}\wedge\StratI{0})$
								\end{tabular}};
	\node (stop) [startstop, below of=dec1, yshift=1.5cm,fill=blue!20] {return $((\specz,\speco),(\assumpz,\Stratz),(\assumpo,\Strato))$};
	\node (gamemid) [below of=start, yshift=0cm] {};

	\draw [arrow,red!80!black] (dec1) -- (gamemid);
	\draw [arrow] (start) -| (game0);
	\draw [arrow] (start) -| (game1);
	\draw [arrow] (game0) -- (algo0);
	\draw [arrow] (game1) -- (algo1);
	\draw [arrow] (algo0) -- (dtemp0);
	\draw [arrow] (algo1) -- (dtemp1);
	\draw [arrow] (temp0) -- (dec1);
	\draw [arrow] (temp1) -- (dec1);
	\draw [arrow,green!80!black] (dec1) -- node[anchor=west] {true} (stop);
	\draw [arrow,red!80!black] (dec1) |- node[anchor=west, yshift=-1.4cm] {false} (game0);
	\draw [arrow,red!80!black] (dec1) |- node[anchor=south] {
		\begin{tabular}{c} 
		$\speci\gets\speci\wedge\LTLalways (\winz\cap\wino)$\\
		$\hspace*{1cm}\wedge\LTLeventually\LTLalways \neg(\conflict_0\cup\conflict_1)$
		\end{tabular}}
		(game1);	
\end{tikzpicture}
\caption{Flowchart illustration of \textsc{Negotiate} (\cref{alg:negotiate}).}\label{fig:negotiateAlgo}
\end{figure}

\begin{algorithm}[t]
	\caption{$\negotiate(\gamegraph,\specz,\speco)$}\label{alg:negotiate}
	\begin{algorithmic}[1]
		\Require  $ \gamegraph=\tup{V, E}$, $\specz=\parity(\priorityI{0})$, 
		$\speco=\parity(\priorityI{1})$, 
		\Ensure modified specifications $(\specz,\speco)$; \csms $(\assump_0,\Strat_0)$, $(\assump_1,\Strat_1)$
		\State $\wini, \conflict_i, \Strati, \assumpi \gets \parityTemp(\gamegraph,\speci,i)$, $\forall i\in\{0,1\}$\label{alg:negotiation:parity}
		\If{$\checkTemplate(\gamegraph, \assumpI{1-i} \wedge \Strati)=\mathtt{true}$, $\forall i\in\{0,1\}$}\label{alg:negotiation:if}
		\State \Return $ (\specz,\speco), (\assumpz , \Stratz) , (\assumpo , \Strato) $ \label{alg:negotiate:returnmid}
		\Else
		\State $\specir \gets \speci\wedge \LTLalways (\winz \cap \wino) \wedge \LTLeventually\LTLalways \neg (\conflict_0 \cup \conflict_1)$ , $\forall~i\in\{0,1\}$ \label{alg:negotiate:addColive}
		\State 	\Return $\negotiate(\gamegraph, \speczr,\specor)$ \label{alg:negotiate:recursion}
		\EndIf
	\end{algorithmic}
\end{algorithm}

Our algorithm, called $\negotiate$, is depicted schematically in \cref{fig:negotiateAlgo} and given formally in \cref{alg:negotiate}. It uses $\parityTemp$ to compute adequately permissive \csms $(\assumpi,\Strati)$ for each $\p{i}$ in its corresponding game $(\gamegraph,\speci)$ locally (\cref{alg:negotiation:parity} in \cref{alg:negotiate}). These \csms  are then checked for compatibility (as in \cref{def:compatible}) via the function $\checkTemplate$ defined in \cref{sec:CheckTemplate}. If \csms are compatible, they define \irmac-specifications (via \cref{thm:csm-iRmaC}) and hence, \cref{prob:main} is solved and $\negotiate$ terminates.

If they are not compatible, existing conflicts need to be resolved as formalized in \cref{sec:resolve}. The required strengthening of both \csms is again done locally by solving games with modified specifications (red arrow looping back in \cref{fig:negotiateAlgo}) again via $\parityTemp$. 

As the resulting new \csms might again be conflicting, this strengthening process repeats iteratively. We prove in \cref{sec:soundness} that there are always only a finite number of negotiation rounds.

\subsection{Checking Compatibility Efficiently}\label{sec:CheckTemplate}
This section discusses how the procedure $\checkTemplate$ checks compatibility of \csms efficiently. Based on \cref{def:compatible}, checking compatibility of two \csms reduces to checking the existence of a strategy that follows the templates $\assumpI{1-i}\wedge\Strati \compt (\safegroup,\colivegroup,\condlivegroup)$ for both $i\in\{0,1\}$.
As our templates are just particular LTL formulas, one can of course use automata-theoretic techniques to check this.
However, given the edge sets $(\safegroup,\colivegroup,\condlivegroup)$ this check can be performed more efficiently as formalized next.

\begin{definition}\label{def:conflict-freetemplate}
	A template $\template\compt\Tuple{\safegroup,\colivegroup, \condlivegroup}$ over game graph $\gamegraph = (V,E)$ is \emph{conflict-free} if %
	\begin{compactenum}[(i)]
		\item every vertex $v$ has an outgoing edge that is neither co-live nor unsafe, i.e., 
		$v\times E(v) \not \subseteq D \cup S$, and
		\item in every live-group $\livegroupSingleN\in\livegroup$ s.t.\ $(\cdot,\livegroup)\in\condlivegroup$, every source vertex $v$ has an outgoing edge in $\livegroupSingleN$ that is neither co-live nor unsafe, i.e.,
		$v\times \livegroupSingleN(v) \not \subseteq D \cup S$.
	\end{compactenum}
	The procedure of checking (i)-(ii) is called $\checkTemplate(\gamegraph, \template) $, which returns \texttt{true} if (i)-(ii) holds, and \texttt{false} otherwise.
\end{definition}
We note that checking (i)-(ii) can be done independently for every vertex, hence $\checkTemplate(\gamegraph, \template)$ runs in linear time $\bigO(n)$ for $n=|V|$.
Intuitively, whenever the existentially quantified edge in (i) and (ii) of \cref{def:conflict-freetemplate} exists, a strategy that alternates between all these edges follows the given template. In addition, this strategy can also be extracted in linear time. This is formalized next.

\begin{proposition}\label{prop:strategyExtraction}
	Given a game graph $\gamegraph = (V,E)$ with conflict-free template
	$\template \compt \Tuple{\safegroup,\colivegroup,\condlivegroup}$ for $\p{i}$, a strategy $\strati$ for $\p{i}$ that follows $\template$ can be extracted in time $ \bigO(m) $, where $ m $ is the number of edges. This procedure is called $ \extract(\gamegraph, \template) $.
\end{proposition}
\begin{proof}
The proof is straightforward by constructing the strategy as follows.
We first remove all unsafe and co-live edges from $\gamegraph$ and then construct a strategy $\strati$ that alternates between all remaining edges from every vertex. This strategy is well-defined
as condition (i) in \cref{def:conflict-freetemplate} ensures that after removing all the unsafe and co-live edges a choice from every vertex remains. Moreover, if the vertex is a source of a live-group edge, condition (ii) in \cref{def:conflict-freetemplate} ensures that  there are outgoing edges satisfying every live-group. Thereby, the constructed strategy indeed follows $\template$. 
\end{proof}

It is worth noting that $\parityTemp$ (and all its variants given in \cref{app:template extraction}) always return conflict-free templates $\assumpI{i}$ and $\Strati$ by construction. Only when combining templates from different players into $\Strati\wedge\assumpI{{1-i}}$ conflicts may arise. However, as conflict-freeness of template $\assumpI{1-i}\wedge\Strati$ implies the existence of a $\p{i}$ strategy following it from \cref{prop:strategyExtraction}, this immediately implies that both \csms are compatible, leading to the following corollary.

\begin{corollary}\label{prop:conflictfree2compatible}
	Given two \csms $(\assumpz,\Stratz)$ and $(\assumpo,\Strato)$ in a game graph $\gamegraph$, if for each $i\in\{0,1\}$ the template $\assumpI{1-i}\wedge\Strati$ is conflict-free, then the two \csms are compatible.
\end{corollary}

We note that the converse of \cref{prop:conflictfree2compatible} is not true, as there can be a strategy following $\assumpI{1-i}\wedge\Strati$ even when the corresponding \csms are not conflict-free. 
However, this does not affect the completeness of our algorithm.
Therefore, we focus our attention on ensuring conflict-freeness rather than compatibility.
Moreover, if such a strategy exists it will be retained by the conflict resolving mechanism of $\negotiate$, introduced next.%

\subsection{Resolving Conflicts}\label{sec:resolve}
Given a conflict in $\assumpI{1-i}\wedge\Strati \compt (\safegroup,\colivegroup,\condlivegroup)$ we now discuss how the modified specifications $\specir$ (as in \cref{alg:negotiate:addColive} of \cref{alg:negotiate}) allows to resolve this conflict in the next iteration. 

For this, first assume that $\colivegroup=\emptyset$. In this case a conflict exists because all available (live) edges are unsafe and should never be taken. Hence, an extracted strategy (via \cref{prop:strategyExtraction}) is not well-defined (i.e., might get stuck in a vertex for which (i) of \cref{def:conflict-freetemplate} is false) or not ensured to be winning (i.e., will not be able to fulfill the liveness obligations in $\livegroup$ if (ii) of \cref{def:conflict-freetemplate} is false). 

In order to ensure strategies to be winning, templates need to be re-computed over a game graph where unsafe edges $e\in\safegroup$ in $\assumpI{1-i}\wedge\Strati$ are removed.
By looking into the details of the computation of $\safegroup$ within $\parityTemp$, we see that unsafe edges always transition from the winning region $\win_i=\team{i}(\assumpI{i}, \StratI{i})$ to its complement $\overline{\win}_i=\neg\win_i$, i.e., every (cooperatively winning) play should never visit states in $\overline{\win}_i$.
We therefore achieve the desired effect by adding the requirement $\LTLalways \neg (\overline{\win}_0\cup\neg\overline{\win}_1)=\LTLalways(\win_0\cap\win_1)$ to the specification, which obviously does not restrict the cooperative winning region, as $\parityTemp$ is ensured to not remove any cooperative solution (due to \cref{item:def:APA:sufficient} in \cref{def:adequate assumption}). %

This intuition generalizes to the case where $\colivegroup\neq\emptyset$ as follows. Here, we need to resolve the game while ensuring that co-live edges $e\in\colivegroup$ are only taken finitely often. In analogy to unsafe edges, co-live edges are computed by $\parityTemp$ s.t.\ they always transition to the set of vertices  $\mathcal{C}_i$ that must only be seen finitely often along a winning play. In addition to $\win_i$, the set $\mathcal{C}_i$ can also be memorized during the computation of $\colivegroup$ within $\parityTemp$ and hence passed to $\checkTemplate$ in \cref{alg:negotiation:parity} of \cref{alg:negotiate}. As for the unsafe-edge case, we can achieve the desired effect for recomputation by adding the requirement $\LTLeventually\LTLalways \neg(\conflict_0\cup\conflict_1)$ to the specification $\speci$ (see \cref{alg:negotiate:addColive} of \cref{alg:negotiate}). Again, this obviously does not alter the cooperative winning region of the game.

\begin{remark}\label{rem:otherinput}
 We note that \cref{alg:negotiate} is slightly simplified, as the objective $\specir$ in \cref{alg:negotiate:addColive} of \cref{alg:negotiate} used as an input to $\negotiate$ in latter iterations, is not a \enquote{plain} parity objective $\parity(\priority)$. As $\parityTemp$ expects a parity game as an input, we need to convert $(\gamegraph,\speczr,\specor)$ into a parity game by a simple reprocessing step. Luckily, both additional specifications can be dealt with using classical steps of Zielonka's algorithm \cite{ZIELONKA1998135}, a well-known algorithm to solve parity games, which is used as the basis for $\parityTemp$. Concretely, we handle the $\LTLalways(\win_0\cap\win_1)$ part of $\specir$, by restricting the game graph $\gamegraph$ to $\win=\win_0\cap\win_1$ and the $\LTLeventually\LTLalways \neg(\conflict_0\cup\conflict_1)$ part by assigning all vertices in $\conflict=(\conflict_0\cup\conflict_1)$ the highest odd priority $2d_i+1$. The correctness of these standard transformations follows from the same arguments as used to prove the correctness of similar steps of the $\parityTemp$ algorithm in \cref{app:template extraction}. 
\end{remark}

\subsection{Properties of $\negotiate$}\label{sec:soundness}
With this, we are finally ready to prove that (i) $\negotiate$ always terminates in a finite number of steps, and (ii) upon termination, the computed \csms indeed provide a solution to \cref{prob:main}.

\smallskip
\noindent\textbf{Termination.}
Intuitively, all local synthesis problems are performed over the same (possibly shrinking) game graph $\gamegraph$. Therefore, there exists only a finite number of templates $\template\compt(\safegroup,\colivegroup,\condlivegroup)$ over $\gamegraph$, which, in the worst case, can all be enumerated in finite time.

\begin{theorem}\label{thm:terminate}
 Given a two-objective parity game $\game=((V,E),\specI{0},\specI{1})$ with $\specI{i}=\paritygame(\priority_i)$, \cref{alg:negotiate} always terminates in $\bigO(n^6)$ time, where $n = \abs{V}$.
\end{theorem}

\begin{proof}
 We prove termination via an induction on the lexicographically ordered sequence of pairs $(\abs{\win},\abs{\win\setminus \conflict})$. %
	As the base case, observe that if $\abs{\win} = 0$ we have that $\win_i=\emptyset$ for at least on \p{i}, implying $\assumpi\compt(\emptyset,\emptyset,\emptyset)$ and $\Strati\compt(\emptyset,\emptyset,\emptyset)$ for this \p{i}. As $\assumpi\wedge\Strat_{1-i}=\Strat_{1-i}$ and $\assump_{1-i}\wedge\Strati=\assump_{1-i}$ in this case, and $\assump_{1-i}$ and $\Strat_{1-i}$ are conflict-free by construction, $\checkTemplate$ returns $\mathtt{true}$ and $\negotiate$ terminates. If instead only $\abs{\win\setminus \conflict}=0$, it follows from \cref{rem:otherinput} that all vertices in $\win$ have highest odd priority, implying that in the next iteration $\win_i = \emptyset$ for both players, hence all templates are empty, i.e., trivially conflict-free, hence $\negotiate$ terminates.
	
	Now for the induction step, suppose $\abs{\win}>0$ and $\abs{\win\setminus \conflict}>0$ in the \emph{previous} iteration. 	If $\assumpz \wedge \Strato$ and $\assumpo \wedge \Stratz$ are conflict-free, $\negotiate$ terminates. 
	Suppose this is not the case. As $\gamegraph$ gets restricted to $\win$ for this iteration (see \cref{rem:otherinput}), unsafe edges can only occur if $\win'\subset\win$ (as they are by construction from $\win'$ to $\neg\win'$ s.t.\ the latter is a subset of $\win$), where $\win'$ is the winning region computed in the \emph{current} iteration. If $\win'=\win$ conflicts need to arise from colive edges. As colive edges are computed by $\parityTemp$ in a subgame that excludes all vertices with the highest odd priority (and therefore all vertices in $\conflict$ due to \cref{rem:otherinput}), the existence of co-live edge conflicts implies the existence of colive edges, which implies that $\abs{\win'\setminus \conflict'}<\abs{\win\setminus \conflict}$. Therefore, $(\abs{\win},\abs{\win\setminus \conflict})$ always reduces (lexicographically) when conflicts occur. Hence, the algorithm terminates by induction hypothesis.
	Furthermore, as each iteration calls $\checkTemplate$ once and $\parityAssump$ twice which runs in $\bigO(\abs{V}^4)$ time, \cref{alg:negotiate} terminates in $\bigO(\abs{V}^6)$ time.
\end{proof}

\smallskip
\noindent\textbf{Soundness.}
While it seems to immediately follow that \irmac-specifications can be \csms that $\negotiate$ outputs, as it only terminates on adequately permissive and compatible \csms, this is only true w.r.t.\ the new game $(\gamegraph,\speczr,\specor)$ which gets modified in every iteration. It therefore remains to show that the resulting \csms induce \irmac-specification for $(G,\specz,\speco)$, which then proves that $\negotiate$ solves \cref{prob:main}.

\begin{theorem}\label{thm:negotiate-irmac}
 Let $((\speczq,\specoq),(\assumpz , \Stratz) , (\assumpo , \Strato))$ be the output of $\negotiate(\gamegraph ,\specz, \speco)$. Then $(\speczp,\specop)$ with $\specip := \assumpI{1-i}\wedge(\assumpI{i}\Rightarrow\speciq)$ are \irmac-specifications for $(\gamegraph, \specI{0},\specI{0})$.
\end{theorem}
\begin{proof}
As \csm $(\assumpi,\Strati)$ is adequately permissive for $\p{i}$ in the game $(\gamegraph,\speci'')$ and as the returned \csms are compatible, by using \cref{thm:csm-iRmaC}, the contracted specifications $(\speczp,\specop)$ for the two-objective game $(\gamegraph,\specz'',\speco'')$ are \irmac-specifications. Hence, 
$\lang(\specz''\wedge\speco'') = \lang(\speczp\wedge\specop)$ and 
$\team{0,1} \specz''\wedge\speco'' = \team{0}\speczp~\cap~\team{1}\specop$.
Hence, in order to prove that \eqref{equ:ceis} holds, it suffices to show that $\specz''\wedge\speco''$ is equivalent to $\specz\wedge\speco$, i.e., $\lang(\specz''\wedge\speco'') = \lang(\specz\wedge\speco)$. This however immediately follows from the fact that \parityTemp computes $\win_i$ and $\conflict_i$ in an adequately permissive manner, i.e., never excluding any cooperative winning play. Thereby, the addition of the terms 
$\LTLalways(\win_0\cap\win_1)$ and
$\LTLeventually\LTLalways \neg(\conflict_0\cup\conflict_1)$ to the specification does not exclude cooperative winning plays either, hence keeping $\lang(\specz\wedge\speco)$ the same in each iteration.
\end{proof}

\smallskip
\noindent\textbf{Decoupled Strategy Extraction.}
By combining the properties of \irmac-specifications with \cref{prop:strategyExtraction}, we have the following proposition which shows that by using templates to formalize \irmac contracts, we indeed fully decouple the strategy choices for both players. %

\begin{proposition}\label{cor:finalresult}
	In the context of \cref{thm:negotiate-irmac}, let $\pi_i$ be a strategy of player $i$ following $\assumpI{1-i}\wedge\Strati$. Then
	\begin{enumerate}[(i)]
	 \item $\strati$ is winning in $(G,\specip)$ from every $v\in\team{0,1}(\specz\cap\speco)$, and 
	 \item the strategy profile $(\pi_0,\pi_1)$ is winning in $(\gamegraph,\spec_0,\spec_1)$. %
	\end{enumerate}
\end{proposition}

\begin{proof}
\textbf{(i)} As the \csm $(\assumpi,\Strati)$ is adequately permissive for $\p{i}$ in the game $(\gamegraph,\speci'')$, by \cref{rem:winningcsm}, the sufficiency condition makes it winning from all vertices in $\team{0,1}\speci''\supseteq\team{0,1}(\specz''\wedge\speco'') = \team{0,1}(\specz\wedge\speco)$.
	Moreover, as $\strati$ follows $\Strati$, by using \cref{rem:winundera}, $\strati$ is winning in the game $(\gamegraph,\assumpi\Rightarrow\speci'')$ from $\team{0,1}(\specz\wedge\speco)$. Hence, every $\strati$-play from $\team{0,1}(\specz\wedge\speco)$ satisfies both $\assumpI{1-i}$ and $\assumpi\Rightarrow\speci''$. 
	As $\specip = \assumpI{1-i}\wedge(\assumpi\Rightarrow\speci'')$, strategy $\strati$ is winning in the game $(\gamegraph,\specip)$ from $\team{0,1}(\specz\wedge\speco)$.

\textbf{(ii)} This now follows directly from (i) and \cref{prop:irmac}.
	\end{proof}

\smallskip
\noindent\textbf{Completeness.} As our final result, we note that as a simple corollary from \cref{prop:strategyExtraction} and \cref{cor:finalresult} follows that whenever a cooperative solution to the original synthesis problem $(\gamegraph,\spec_0,\spec_1)$ exists, we can extract a winning strategy profile from the \csms computed by $\negotiate$.
	\begin{corollary}
	 	In the context of \cref{thm:negotiate-irmac}, for any vertex $v$ from which there exists a winning strategy profile $(\pi'_0,\pi'_1)$ for the two-objective parity game $(\gamegraph,\spec_0,\spec_1)$, there exist strategies $\pi''_i$ from $v$ following $\assumpI{1-i}\wedge\Strati$ for both $i\in\{0,1\}$.
	\end{corollary}

\section{Experimental Evaluation}\label{sec:experiements}
To demonstrate the effectiveness of our approach, we conducted experiments using a prototype tool \toolname\footnote{Repository: \url{https://github.com/satya2009rta/cosmo}} that implements the negotiation algorithm (\cref{alg:negotiate}) for solving two-objective parity games. All experiments were performed on a computer equipped with an Apple M1 Pro 8-core CPU and 16GB of RAM.

\subsection{Factory Benchmark}\label{sec:experiments:factory}
Building upon the running example in \cref{sec:intro}, we generated a comprehensive set of 2357 factory benchmark instances. These instances simulate two robots, denoted as $R_1$ and $R_2$, navigating within a maze-like workspace.
We used four parameters, i.e., size of the maze $x\times y$, number of walls $w$, and maximum number of one-way corridors $c$.
First, we consider the B\"uchi objective that robots $R_1$ and $R_2$ should visit the upper-right and upper-left corners, respectively, of the maze infinitely often, while ensuring that they never occupy the same location simultaneously and do not bump into a wall. Second, we consider the parity objectives from \cref{exp:pens:1}.
Further details regarding the generation of these benchmark instances can be found in \cref{app:experiments-benchmarks}.
An illustration of one such benchmark is depicted in \cref{fig:experiments} (left).

\begin{figure}
    \centering
    \includegraphics[width=0.4\linewidth]{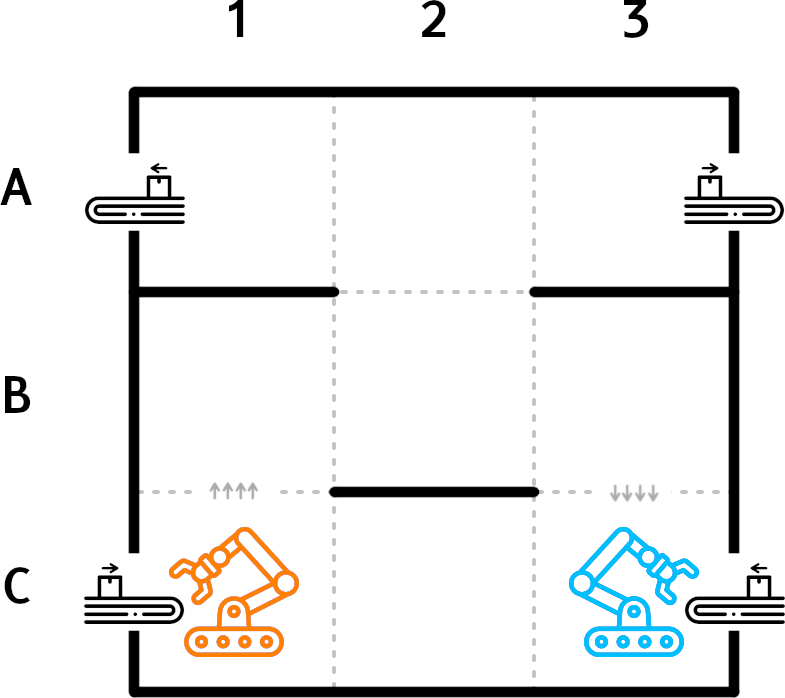}
    \def\svgwidth{0.5\linewidth}
	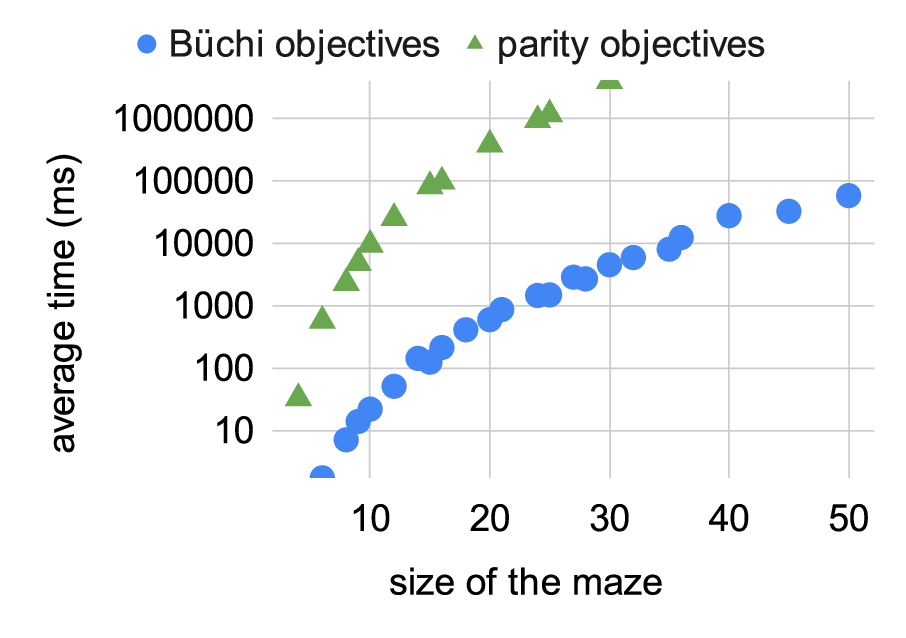
    \caption{Left: Example of a factory benchmark with parameters $x=3$, $y=3$, $w=3$, and $c=2$. Solid lines denote walls, little up- and downward pointing arrows indicate one-way corridors. Right: Data points for factory benchmarks with \buchi objectives (blue circles) and parity objectives (green triangles) describing average execution time over all instances with the same grid size. The $y$-axis is given in log-scale.}
    \label{fig:experiments}
    \vspace{-0.5cm}
\end{figure}

\smallskip
\noindent\textbf{Experimental Results.}\label{sec:evaluation:factory}
In a first set of experiments, we ran our tool on all the factory benchmarks instances and plot all average run-times per grid-size (but with varying parameters for $c$ and $w$) in \cref{fig:experiments} (right). 
We see that \toolname takes significantly more time for parity objectives compared to \buchi objectives.  
That is because computing templates for \buchi games takes linear time in the size of the games whereas the same takes biquadratic time for parity games (see \cref{app:template extraction}).
Furthermore, the templates computed for \buchi objectives do not contain co-liveness templates, and hence, they do not raise conflicts in most cases. However, templates for parity objectives contain all types of templates and hence, typically need several rounds of negotiation. %

In a second set of experiments, we compared the performance of \toolname, with the related tool\footnote{Unfortunately, a comparison with the only other related tool \cite{Bosy-based-contract} which allows for parity objectives was not possible, as we were told by the authors that their tool became incompatible with the new version of BoSy and is therefore currently unusable.} \texttt{agnes} implementing the contract-based distributed synthesis method from \cite{AGDistributed}. Unfortunately, \texttt{agnes} can only handle B\"uchi specifications and resulted in segmentation faults for many benchmark instances.
We have therefore 
only report computation times for all instances that have not resulted in segmentation faults.

The experimental results are summarized in \cref{fig:experiments:bars}-\ref{fig:experiments:agnes}.
As \toolname implements a complete algorithm, it provably only concludes that a given benchmark instance is unrealizable, if it truly is unrealizable, i.e., for $1.67\%$ of the considered $120$ instances. \texttt{agnes} however, concludes unrealizability in $36,67\%$ of its instances (see \cref{fig:experiments:bars} (left)), resulting an many false-negatives. Similarly, as \toolname is ensured to always terminate, we see that all considered instances have terminated in the given time bound. While, \texttt{agnes} typically computes a solution faster for a given instance (see \cref{fig:experiments:agnes} (left)), it enters a non-terminating negotiation loop in $13,34\%$ of the instances (see \cref{fig:experiments:bars} (right)). This happens for almost all considered grid sizes, as visible from \cref{fig:experiments:agnes} (right) where all non-terminating instances are included in the average after being mapped to $300s$, which was used as a time-out for the experiments. 

\begin{figure}[]
    \centering
    \def\svgwidth{0.49\linewidth}
	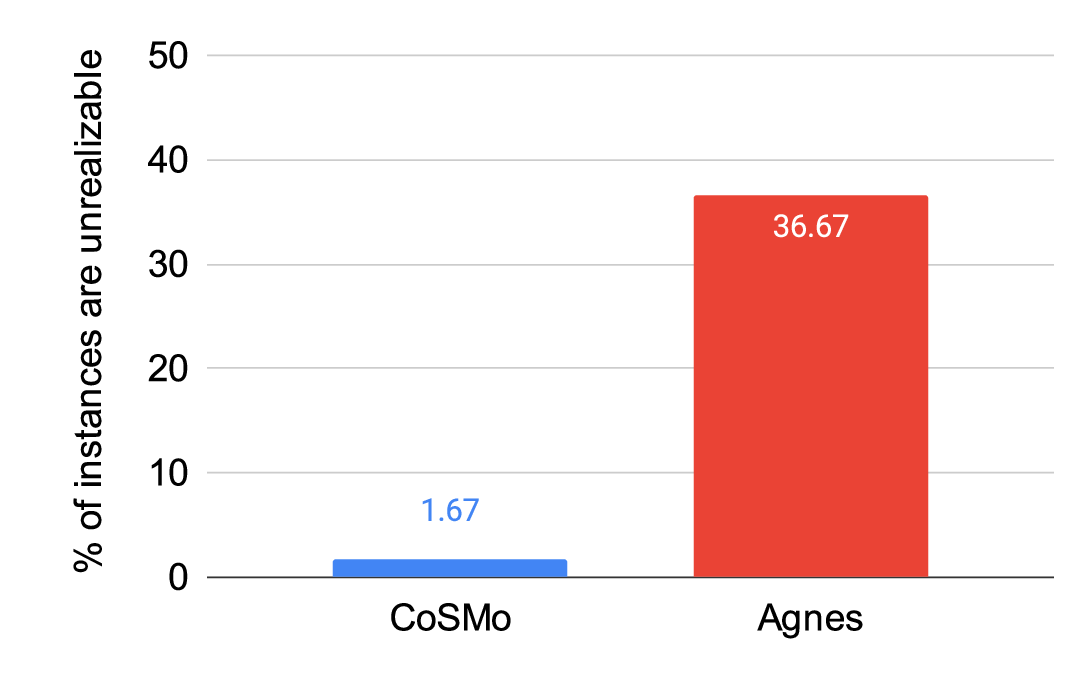
\hfill
	\def\svgwidth{0.49\linewidth}
	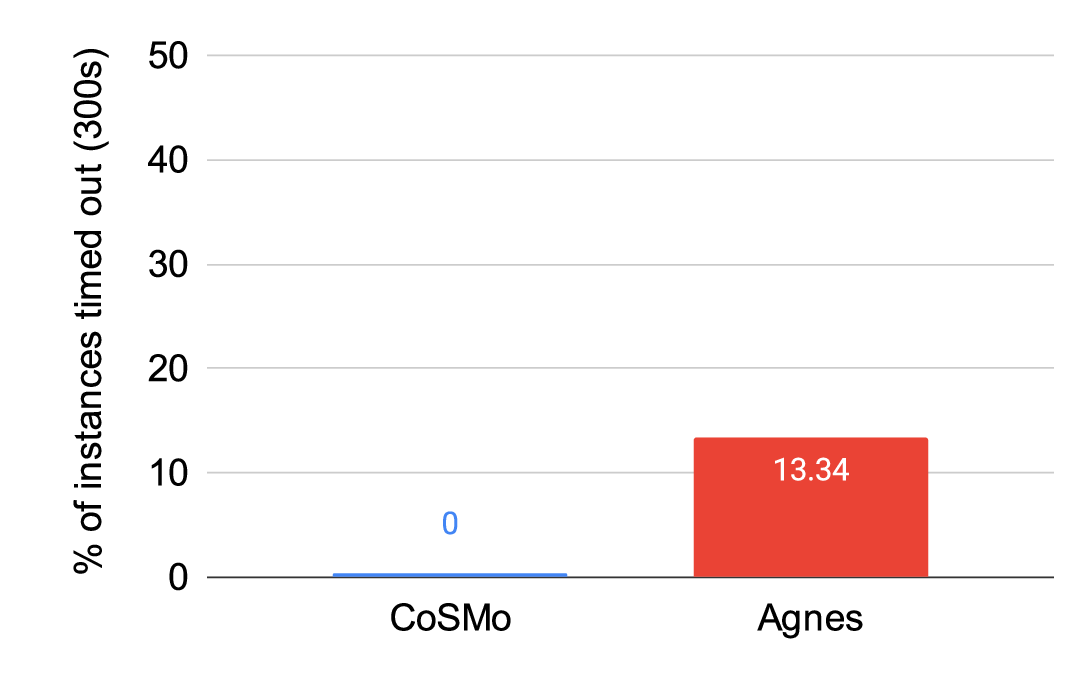
    \caption{Left: Percentage of instances on which the respective tool reports unrealizability after termination. Right: Percentage in which the respective tool does not terminate. Both numbers are mutually exclusive.}
    \label{fig:experiments:bars}
    \vspace{-0.5cm}
\end{figure}

While our experiments show that \texttt{agnes} outperforms \toolname in terms of computation times when it terminates on realizable instances (see \cref{fig:experiments:agnes} (left)), it is unable to synthesize strategies either due to conservatism or non-termination in almost $50\%$ of the considered instances (in addition to the ones which returned segmentation faults and which are therefore not included in the results). In addition to the fact that \texttt{agnes} can only handle the small class of Büchi specifications while \toolname can handle parity objectives, we conclude that \toolname clearly solves the given synthesis task much more satisfactory.

\begin{figure}[]
    \centering
    \def\svgwidth{0.49\linewidth}
	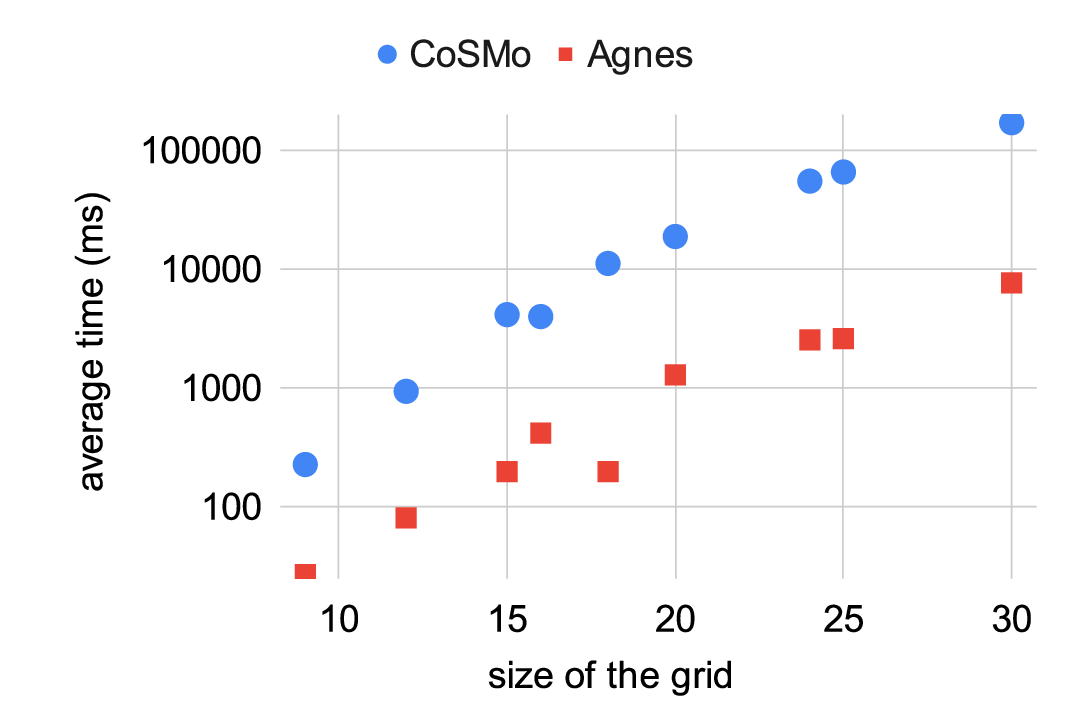
\hfill
    \def\svgwidth{0.49\linewidth}
	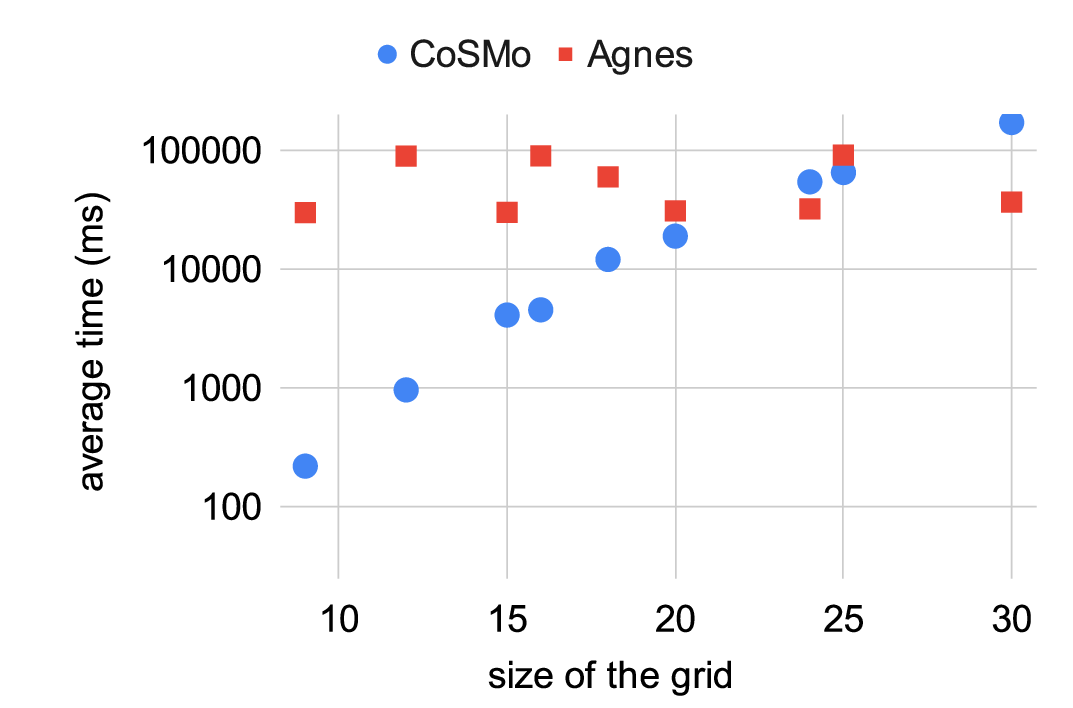
    \caption{Average computation times over all instances with the same grid size for \toolname (blue circles) and \texttt{agnes} (red squares) without timed-out instances (left) and with timed-out instances mapped to the time-out of $300s$ (right). The y-axis is given in log-scale.}
    \label{fig:experiments:agnes}
    \vspace{-0.5cm}
\end{figure}

\newcommand{\compose}{\textsc{Compose}}
\newcommand{\assumpsafe}{\colred safe}
\newcommand{\assumpgrlive}{\colred grlive}
\newcommand{\assumpdep}{\colred dep}
\newcommand{\edges}{\colred edges}

\subsection{Incremental Synthesis and Negotiation}\label{sec:composition}
While the previous section evaluates our method for a single, static synthesis task, we want to now emphasize the strength of our technique for the online adaptation of strategies. 
To this end, we assume that \cref{alg:negotiate} has terminated on the input $(\gamegraph,\spec_0,\spec_1)$ and compatible \csm's $(\assump_0,\Strat_0)$ and $(\assump_1,\Strat_1)$ have been obtained. Then a new parity objective $\spec_i'$ over $G$ arrives for component $i$, for which additional \csm $(\assump'_i,\Strat'_i):=\parityTemp(\gamegraph, \spec'_i)$ can be computed. 
It is easy to observe that if $(\assump'_i,\Strat'_i)$ does not introduce new conflicts, no further negotiation needs to be done and the \csm of component $i$ can simply be updated to $(\assump_i\wedge\assump'_i,\Strat_i\wedge\Strat_i')$. Otherwise, we simply re-negotiate by running more iterations of \cref{alg:negotiate}. 

We note that, algorithmically, this variation of the problem requires solving a chain of \emph{generalized} parity games, i.e., a parity game with a conjunction of a finite number of parity objectives. We therefore compare the performance of \toolname on such synthesis problems\footnote{The details about the benchmarks can be found in \cref{app:experiments-benchmarks}.} to the best known solver for \emph{generalized} parity games, i.e, \genziel from \cite{chatterjee2007:generalizedparitygames} (implemented by  \cite{Buryere2019:partialsolvers}). Similar to our approach, \genziel is complete and based on Zielonka's algorithm. However, it solves one \emph{centralized} cooperative game for the conjunction of all players objectives. %

\begin{figure}[]
\centering
\def\svgwidth{0.7\linewidth}
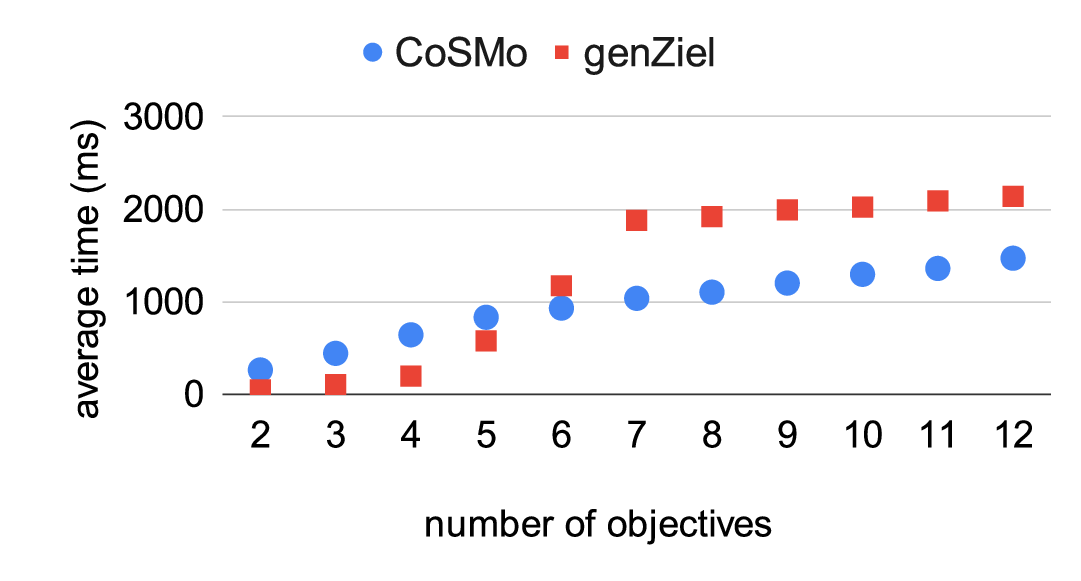
\vspace{-0.4cm}
    \caption{Experimental results over $2244$ games when new parity objectives are added \emph{incrementally one-by-one}. Data points give the average execution time (in ms) over all instances with the same number of parity objectives for \toolname (blue circles) and \genziel \cite{chatterjee2007:generalizedparitygames} (red squares).}
    \label{fig:experiment:incremnetal}
\end{figure}

\begin{figure}[]
\centering
\def\svgwidth{0.49\linewidth}
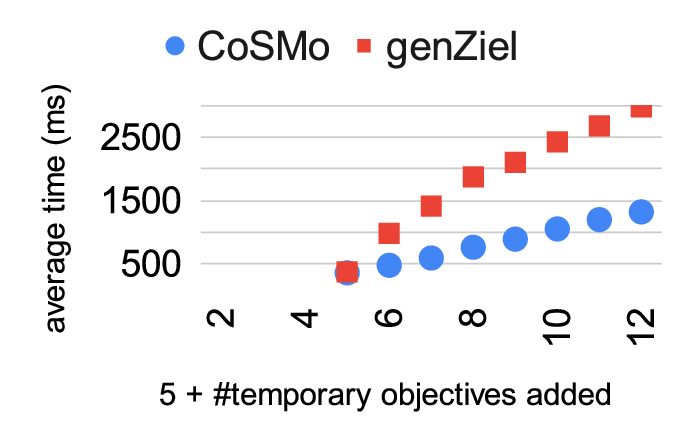
\def\svgwidth{0.49\linewidth}
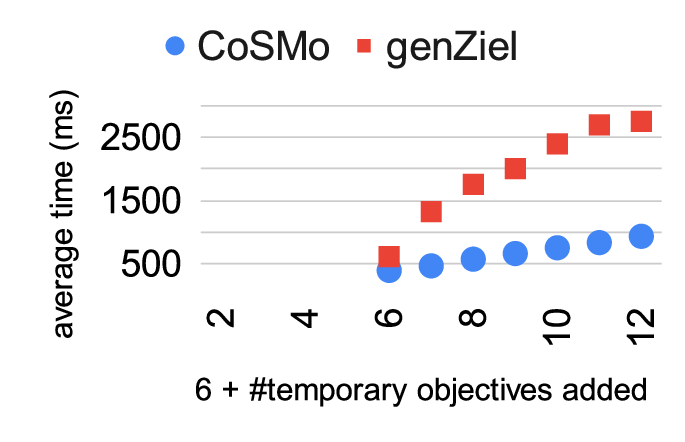
\vspace{-0.8cm}
    \caption{Variation of the experiment in \cref{fig:experiment:incremnetal} with either $5$ (left) or $6$ (right) long-term objectives.}
    \label{fig:experiment:temp}
 \vspace{-0.5cm}
\end{figure}

\smallskip
\noindent\textbf{Comparative evaluation.}
\cref{fig:experiment:incremnetal} shows the average computation time of \genziel and \toolname when objectives are \emph{incrementally one-by-one}, i.e., the game was solved with $\ell$ objectives, then  one more objective was added and the game was solved it again. We see that for a low number of objectives, the negotiation of contracts in a distributed fashion by \toolname adds computational overhead, which reduces when more objectives are added. However, as more objectives are added the chance of the winning region to become empty increases. This gives \genziel an advantage, as it can detect an empty winning region very quickly (due to its centralized computation).
In order to separate the effect of (i) the increased number of re-computations and (ii) the shrinking of the winning region induced by an increased number of incrementally added objectives, we conducted a section experiment where added objectives are allowed to disappear again after some time. %
Here, we consider benchmarks with a fixed number of long-term objectives, and iteratively add just one \emph{temporary} objective at a time.
The results are summarized in \cref{fig:experiment:temp} when the number of long-term objectives are $5$ (left) and $6$ (right). We see that in this scenario \toolname clearly outperforms \genziel, while performing computations in a distributed manner and returning strategy templates.

 \bibliographystyle{abbrv}
 \bibliography{main.bib}

 \begin{appendix}
 
\section{Details on the Computation of Adequately Permissive \csms}\label{app:template extraction}

The insight that we exploit in this section, is the fact that the algorithmic computation of adequately permissive (and hence, also winning) \csm's is similar to the computation of adequately permissive assumptions for parity games from \cite{SIMAssumptions22}, with particular modifications to extract both assumption and winning strategy templates \emph{at the same time}. %

\subsection{Set Transformers} 
 We use some set transformation operators in the algorithms to compute strategy templates. 
Let $ \gamegraph=(V=\vertexz\cupdot \vertexo, E) $ be a game graph, $ U\subseteq V $ be a subset of vertices, and $ a\in \{0,1\} $ be the player index. Then we define a predecessor, \emph{controllable predecessor} $ \cprea{\gamegraph}{U} $ as
\begin{eqnarray}
	 \cprea{\gamegraph}{U}=&\bigcup
	\begin{cases}
		\{v\in V^a\mid \exists (v,u)\in E.~u\in U \}\\
		 \{v\in V^{1-a}\mid \forall (v,u)\in E.~u\in U  \}
	\end{cases}\\
\cpre{\gamegraph}{U}{a,1} =& \cpre{\gamegraph}{U}{a}\cup U\\
\cpre{\gamegraph}{U}{a,i}=&\cpre{\gamegraph}{\cpre{\gamegraph}{U}{a,i-1}}{a} \cup \cpre{\gamegraph}{U}{a,i-1}
\end{eqnarray}
where $ i\geq 1 $.
Intuitively, the operators $ \textsf{cpre}^a_{\gamegraph}(U) $ and $\cpre{\gamegraph}{U}{a,i}$ compute the set of vertices from which $ \p{a} $ can force visiting $ U $ in \emph{at most} \emph{one} and $i$ steps respectively. 
In the following, we introduce the attractor operator $ \textsf{attr}^a_{\gamegraph}(U) $ that computes the set of vertices from which $ \p{a}$ can force at least a single visit to $ U $ in \emph{finitely many but nonzero}\footnote{In existing literature, usually $ U\subseteq\mathsf{attr}^a(U) $, i.e.\ $\attra{}{U}$ contains vertices from which $U$ is visited in zero steps. We exclude $U$ from $\attra{}{U}$ for a minor technical reason.} steps: 
\begin{eqnarray}
	\attra{\gamegraph}{U} =&\big(  \bigcup_{i\geq 1} \textsf{cpre}^{a,i}{(U)} \big)\backslash U
\end{eqnarray}

It is known that $ \attra{\gamegraph}{U} $ can be computed in finite time, and by $ \computeAttr_a(\gamegraph, U) $, we denote the procedure that takes a game graph $ \gamegraph $ and a set of vertices $ U \subseteq V $, and outputs $ \attra{\gamegraph}{U} $, and the live groups $ \{(u,v)\mid u\in\cpre{\gamegraph}{U}{a,i}\text{ and } v\in\cpre{\gamegraph}{U}{a,i-1}, i>0  \} $.

In the following, we show the algorithms to compute assumptions on the other player, and a strategy template under the assumption for self, for one of the most simple but useful objective.

\subsection{Safety Games}

A safety game is a game $\game=(\gamegraph,\spec)$ with $\spec=\square U$
for some $U\subseteq V$, and a play fulfills $\spec$ if it never leaves $U$.
It is well-known that an assumption satisfying the properties of being adequately permissive, and a winning strategy template for safety games disallow every move that leaves the cooperative winning region in $\gamegraph$. This is formalized in the following theorem. 

\begin{theorem}[\cite{chatterjee2008environment,Klein2015:mostGeneralController}]\label{thm:safety assumption}
Let $\game=(\gamegraph,\square U)$ be a safety game, $Z^*:=\nu Y. U\cap \pre{}{Y}$, and $ \safegroup_i = \Set{(u,v)\in E\mid \left(u\in \vertexi\cap Z^*\right) \wedge \left(v \notin Z^*\right)}$. 
Then $\team{0,1}\square U=Z^*$ is the cooperative winning region.

Furthermore, for both $i\in\{0,1\}$ the tuple $ (\assump_i,\spec_i) $ is an adequately permissive \csm for \p{i}, where $ \assump_i=(\safegroup_{1-i},\emptyset,\emptyset) $ and $ \spec_i= (\safegroup_i,\emptyset,\emptyset) $.  
\end{theorem}

We denote by $\computeSafe(\gamegraph,U,i)$ the algorithm computing $(\safegroup_i,\safegroup_{1-i})$ via the fixed point $\nu Y. U\cap \pre{}{Y}$ used within \cref{thm:safety assumption}. This algorithm runs in time $ \bigO(m) $, where $ m=|E|$. 

\subsection{\buchi games}
A \buchi game is a game $\game=(\gamegraph,\spec)$ where $\spec=\square\lozenge I$ for some $I\subseteq V$.
Intuitively, a play is winning for a \buchi objective if it visits the vertex set $I$ infinitely often.

\begin{algorithm}[t]
	\caption{\buchiTemp}
	\label{alg:compute buechi assumption}
		\begin{algorithmic}[1]
			\Require $ \gamegraph=\tup{V=\vertexz\cupdot \vertexo, E}$, $I\subseteq V$, and $i\in\{0,1\}$
			\Ensure \csm $(\assump_i,\Strat_i)$
			\State $Z^* \gets \solveBuchi_{0,1}(\gamegraph,I)$\label{algo:coop buechi computation}
			\State $(\safegroupS, \safegroupA)\gets  \computeSafe(\gamegraph,Z^*,i)$
			\State $\gamegraph\gets \gamegraph|_{Z^*}, I \gets I\cap Z^*$\label{algo:restricted buechi game}
			\State $ (\livegroupS, \livegroupA) \gets $\computeLive$(\gamegraph,I,i) $
			\State \Return $((\safegroupS,\emptyset,\livegroupS),(\safegroupA,\emptyset,\livegroupA))$ 
			
			\Statex
			\Procedure {\computeLive}{$\gamegraph,I,i$}
			\State $\livegroupS\gets\emptyset; \livegroupA\gets\emptyset$
			\State $ U\gets I $
			\While{$ U\not=V $}
			\State $ (W_{attr}, \livegroup')\gets \computeAttr_i(\gamegraph, U) $\label{algo:buechi:compute attr}
			\State $ \livegroupS\gets \livegroupS\cup \livegroup' $ \label{algo:buechi:add attr i}
			\State $ U\gets U\cup W_{attr} $\label{algo:buechi:add attr}
			\State $ C\gets \cpre{\gamegraph}{U}{1-i} $\label{algo:buechi:compute front}
			\State $ \livegroupA\gets \livegroupA \cup \{\{ (u,v)\in E\cap (C\times U) \}\} $\label{algo:buechi:add assump live group}
			\State $ U\gets U\cup C $\label{algo:buechi:add front}
			\EndWhile
			\State \Return $ (\livegroupS,\livegroupA) $
			\EndProcedure
		\end{algorithmic}
\end{algorithm}

Let us present \cref{alg:compute buechi assumption} that computes assumption and strategy templates for \buchi games.
We observe that the algorithm only outputs safety and group-liveness templates. The reason lies in the intrinsic nature of the \buchi objective. $ \p{i} $ needs to visit $ I $ infinitely often. Since there is no restriction on how frequently should they visit $ I $, it suffices to always eventually make progress towards $ I $, and this behavior is precisely captured by live groups. The algorithm exploits this observation and first finds the set of vertices from where $ \p{i} $ can visit $ I $ infinitely often, possibly with help from $ \p{1-i} $. Then all the edges going out from these vertices to the losing vertices becomes unsafe for both players. Now, we need to find ways to find progress towards the goal vertices. To this end, in \cref{algo:buechi:compute attr}, the algorithm finds the vertices from which $ \p{i} $ can visit $ I $ without any help from $ \p{1-i} $, while finding the live groups to facilitate the progress. Then in \cref{algo:buechi:compute front}, the algorithm finds the vertices from which $ \p{1-i} $ needs to help, and forms another live group from these vertices towards the goal vertices. This procedure continues until all the vertices have been covered, while making the covered vertices the new goal.

\begin{restatable}{theorem}{buechiAssumption}\label{thm:Buechi assumptions}

	Given a game graph $ \gamegraph=\tup{V=\vertexz\cupdot \vertexo, E} $ with \buchi objective $ \spec_i=\square\lozenge I $ for $ \p{i} $, \cref{alg:compute buechi assumption} terminates in time $ \bigO(m) $, where $ m $ is the number of edges in the graph. Moreover, $ (\assump_i,\spec_i)= \buchiTemp(\gamegraph, I, i)$ is an adequately permissive \csm for \p{i}.
\end{restatable}
\begin{proof}
	We first show that the algorithm terminates. We show that the procedure $ \computeLive $ terminates. Since in \cref{algo:restricted buechi game}, the game graph is restricted to cooperative \buchi winning region $ Z^* $, we need to show that in the procedure, $ U=V=Z^* $ eventually. Let $ U_l $ be the value of $ U $ after the $ l$-th iteration of $ \computeLive (\gamegraph,I,i)  $, with $ U^0=I $. Since vertices are only added to $ U $ (and never removed) and there are only finitely many vertices, $ U_0\subseteq U_1\subseteq\ldots\subseteq U_{m-1}=U_m $ for some $ m\in\N $. 
	
	Since the other direction is trivial, we show that $ Z^*\subseteq U_m $. Suppose this is not the case, i.e. $ v\in Z^*\backslash U_m $. Since $ v\in Z^* $, both players cooperately can visit $ I $ from $ v $. Then there is a finite path $ \play = v=v_0v_1\cdots v_k $ for $ v_k\in I $. But since $ I=U_0\subseteq U_m $, but $ v\not \in U_m $, $ \play $ enters $ U_m $ eventually. Let $ l $ be the highest index such that $ v_l\not\in U_m $ but $ v_{l+1}\in U_m $.
	
	Then if $ v_l\in \vertexi $, it would be added to $ U $ in \cref{algo:buechi:add attr} of $ (m+1) $-th iteration, i.e. $ U_m\not=U_{m+1} $. Else if $ v_l\in \vertexI{1-i} $, it would be added to $ U $ in \cref{algo:buechi:add front} of $ (m+1) $-th iteration since $ v_{l+1}\in U_m $, i.e. $ U_m\not=U_{m+1} $. In either case, we get a contradiction. Hence, $ v\in U_m $, implying $ Z^*=U_m $.  Hence the procedure \computeLive, and hence \cref{alg:compute buechi assumption}, terminates.
	
	We now show that the \csm obtained is adequately permissive for \p{i}. 
	
	\begin{inparaitem}[$\blacktriangleright$]
		\noindent \item \textbf{(ii)} Implementability: We first show that $ \team{1-i}\assump = V $, but since the forward containment is trivial, we show $ \team{1-i}\assump\supseteq V $. We note that in \cref{algo:buechi:compute front}, $ C\subseteq \vertexI{1-i} $: since if $ v\in \vertexi\cap C $, then there is an edge from $ v $ to $ U $, and hence $ v\in U $ already by \cref{algo:buechi:compute attr} and \cref{algo:buechi:add attr}. Then since the source of live groups (which are only added in \cref{algo:buechi:add assump live group}) is a subset of $ \p{1-i} $'s vertices, then if $ \p{1-i} $ plays one of these live group edges infinitely often, when the sources are visited infinitely often, $ \p{i} $ can not falsify it, and hence $ \team{1-i}\assump\supseteq V $. 
		
		Analogously, $ \team{i}\Strat \supseteq V $, since in Line \ref{algo:buechi:add attr i}, live group edges are added only from $ \p{i}$'s vertices, by definition of $ \computeAttr_i $.
		
		\noindent \item \textbf{(i)} Sufficiency: Again, let $ U_l $ and $ m $ be as defined earlier. Define $ X_l\coloneqq U_l\backslash U_{l-1} $ for $ 1\leq l\leq m $, and $ X_0=U_0=I $. Then every vertex $ v\in Z^* $ is in $ X_l $ for some $ l\in [0;m] $.
		
		Consider the strategy $ \strati$ for $ \p{i} $: at a vertex $ v\in V_i\cap X_l $, she plays the $ \textsf{attr}^i $ strategy to reach $ U_{l-1} $, and for other vertices, she plays arbitrarily. It is easy to observe that $ \strati $ follows $ \spec $. Then we show that for any $ \p{1-i} $ strategy $ \stratj $ following $ \assump $, the resulting $ \strati\stratj $-play $ \play $ from any $ v_0\in Z^* $ belongs to $ \lang(\spec) $.

		Let $ v_0\in Z^*=\team{0,1}\square\lozenge I $ (from \cref{algo:coop buechi computation}). Let $ \stratj $ be an arbitrary strategy that follows $\assump$, and $ \play=v_0v_1\ldots $ be an arbitrary $ \strati\stratj $-play. Then $ \play\in \lang(\assump) $. It remains to show that 
		$ \play \in \lang(\spec) $.
		
		Suppose $ \play \not\in \lang(\spec) $, i.e. $ inf(\play)\cap I=\emptyset $. Note that $ \play $ never leaves $ Z^* $ due to safety assumption template. Then consider the set $ R $ of vertices which occur infinitely often in $ \play $. Let $ 0\leq k\leq m $ be the least index such that $ R\cap X_{k} \not=\emptyset$. From the assumption, $ k>0 $. Let $ v\in R\cap X_k $.
		
		If $ v\in \vertexi $, by the definition of $ \strati $, every time $ \play $ reaches $ v $, it must reach $ U_{k-1} $, contradicting the minimality of $ k $. Else if $ v\in \vertexI{1-i} $, then by the definition of $ \livegroupA $, infinitely often reaching $ v $ implies infinitely often reaching $ \attr{i}{U_{k-1}}{} $. But again the play visits $ U_{k-1} $ by arguments above, giving the contradiction.
		
		In any case, we get a contradiction, implying that the assumption is wrong. Hence, $ \play \in \lang(\spec)$, and $ v_0\in\team{i}(\assump, \Strat) $.

		\noindent \item \textbf{(iii)} Permissiveness: 
		Now for the permissiveness of the \csm, let $ \play\in \lang{(\spec)} $. Suppose that $ \play\not\in\lang{(\assump)} $. 
		
		Case 1: If $ \play\not\in\lang{(\templatesafe(\safegroupA))} $, then some edge $ (v,v')\in \safegroupA $ is taken in $ \play $. Then after reaching $ v' $, $ \play $ still satisfies the \buchi condition. Hence, $ v'\in Z^*=\team{0,1}\square\lozenge I $, but then $ (v,v')\not\in S^a $, which is a contradiction.
		
		Case 2: If $ \play\not\in\lang(\templategrlive(\livegroupA)) $, then $ \exists \livegroupSingle\in \livegroupA, $ such that $ \play $ visits $ src(\livegroupSingle)=C_l $ (for the value of C after $ l $-th iteration) infinitely often, but no edge in $ \livegroupSingle $ is taken infinitely often. Then since the edges in $ \livegroupSingle $ lead to $  \attr{i}{}{U_{l-1}} $, the play must  stay in either $ C_l $ or goes to $ U^{k}\setminus U^l $ for some $ k>l+1 $. In the first case, since $ I\cap Z^*\subseteq U^0 $, $ \play\not\in\spec $, which would be a contradiction. On the other hand, in the second case, after going to $ U^k\setminus U^l $, $ \play $ has an edge going from some $ v\in Z^*\backslash U^{k-1} $ to some $ v'\in U^l $ (else $ I\cap Z^*\subseteq U^0\subseteq U^l $ can not be reached). But then $ v $ would be added to $ U^{k+1} $, which contradicts to the fact that $ k>l+1 $.
		In either case, we get a contradiction, so $ \play\in\lang(\assump)$. %

		\noindent \item \textit{Complexity analysis.}
		The computation of cooperative winning region can be done in time linear in number of edges, i.e. $ \bigO(m) $. The procedure \computeLive takes $ \bigO(m) $ time. Hence, resulting in time linear in number of edges in the game graph. %
	\end{inparaitem}
\end{proof}

\subsection{\cobuchi games}
A \cobuchi game is a game $\game=(\gamegraph,\spec)$ where $\spec=\lozenge\square I$ for some $I\subseteq V$.
Intuitively, a play is winning for a \cobuchi objective if it eventually stays in $ I $.

\begin{algorithm}[h]
	\caption{\cobuchiTemp}
	\label{alg:compute cobuechi assumption}
		\begin{algorithmic}[1]
			\Require $ \gamegraph=\tup{V=\vertexz\cupdot \vertexo, E}, I\subseteq V$, and $i\in\{0,1\}$
			\Ensure \csm $(\assump_i,\Strat_i)$
			\State $Z^* \gets \solveCobuchi_{0,1}(\gamegraph,I)$\label{algo:coop cobuechi computation}
			\State $(\safegroupS, \safegroupA)\gets  \computeSafe(\gamegraph,Z^*,i)$
			\State $\gamegraph\gets \gamegraph|_{Z^*}, I \gets I\cap Z^*$\label{algo:restricted cobuechi game}\Comment{All vertices are cooperatively \cobuchi winning}
			\State $ (\colivegroupS, \colivegroupA) \gets $\computeCoLive$ (\gamegraph,I,i) $
			\State \Return $((\safegroupS,\colivegroupS,\emptyset),(\safegroupA,\colivegroupA,\emptyset))$ 
			
			\Statex
			\Procedure {\computeCoLive}{$\gamegraph,I,i$}
			\State $ U\gets \solveSafety_{0,1}(\gamegraph, I) $\Comment{$ U\subseteq I $}
			\State $ D\gets (U\times V\backslash U)\cap E $\label{alg:colive:step:add colive out of safe region}
			\While{$ U\not=V $}
			\State $ D\gets D\cup ((\pre{\gamegraph}{U}\times \pre{\gamegraph}{U})\cap E)  $\label{alg:colive:step:add colive edge in pre}
			\State $ U\gets U\cup \pre{\gamegraph}{U} $
			\State $ D\gets D\cup ((U\times V\backslash U)\cap E) $\label{alg:colive:step:add colive edge out of pre}
			\EndWhile
			\State $ \colivegroupS\gets \{e\in D\mid src(e)\in \vertexi \} $
			\State $ \colivegroupA\gets \{e\in D\mid src(e)\in V^{1-i} \} $
			\State \Return $ (\colivegroupS,\colivegroupA) $
			\EndProcedure
		\end{algorithmic}
\end{algorithm}

Let us present \cref{alg:compute cobuechi assumption} that computes assumption and strategy templates for \cobuchi games.
The algorithm again only outputs two kinds of templates for assumption and strategy templates: safety and co-liveness. While \buchi requires visiting certain vertices infinitely often, dually, \cobuchi requires avoiding a certain set of vertices eventually (or equivalently staying in the complement eventually). Again, there is no restriction on when exactly a play should stop visiting the said vertices. This paves the way for the application of co-live templates which restrict taking some edges eventually (or equivalently, disallows going away from the goal set of vertices eventually). \cref{alg:compute cobuechi assumption} starts by finding the unsafe edges for both players as in the \buchi case. Then among the vertices from which $ \p{i} $ can satisfy the \cobuchi objective (possibly with help from $ \p{1-i} $), the algorithm finds the vertices from which the play does not even need to leave $ I $ (i.e, $ U $). Clearly, these vertices belong to $ I $, and this is set of vertices where the play should eventually stay. Hence, in \cref{alg:colive:step:add colive out of safe region}, the edges going out of $ U $ are marked co-live, ensuring the play does not always leave this set. Then we need to find vertices from which $ U $ may be reached in one step, and mark the edges going away from $ U $ as co-live. This procedure is repeated till all the vertices are covered. This allows the players to choose the strategies that ensure that the play does not always go away from $ U $ (forcing it to go towards $ U $), and once the play is in $ U $, it stays there eventually.

\begin{restatable}{theorem}{coBuechiAssumption}\label{thm:coBuechi assumptions}
	Given a game graph $ \gamegraph=\tup{V=\vertexz\cupdot \vertexo, E} $ with \cobuchi objective $ \spec_i=\lozenge\square I $ for $ \p{i} $, \cref{alg:compute cobuechi assumption} terminates in time $ \bigO(m) $, where $ m $ is the number of edges. Furthermore, $ (\assump_i,\Strat_i):= \cobuchiTemp(\gamegraph,I,i)$ is an adequately permissive \csm for $ \p{i} $. 
\end{restatable}
\begin{proof}
	We first show that the algorithm terminates. We show that the procedure $ \computeCoLive $ terminates when all the vertices of the game graph are cooperatively winning for the \cobuchi objective $ \speci=\lozenge\square I $, since we restrict the graph to the cooperative winning region in \cref{algo:restricted cobuechi game}. We claim that $ U=V=Z^* $, eventually.
	
	Let $ U_l $ be the value of the variable $ U $ after $ l $-th iteration of the while loop, with $ U_0= \solveSafety_{0,1}(\gamegraph, I) $. Since vertices are only added in $ U $, $ U_0\subseteq U_1\subseteq \ldots \subseteq U_k=U_{k+1} $ for some $ k\in \N $. Suppose $ V\not\subseteq U_k $, then there exists $ v\in V\backslash U_k $. Since $ v\in Z^* $, there is a $ \play = vv_1v_2\ldots $ from $ v $ to $ U_0 $ and stays there forever. Then consider the largest index $ m $ such that $ v_m\not\in U_k $, but $ v_{m+1}\in U_k $. Note that this index exists because $ U_0\subseteq U_k $. But then $ v_m\in \pre{\gamegraph}{U_k} $, and hence $ U_k\not= U_{k+1} $, which is a contradiction. Hence, our assumption that $ v\not \in U_k $ is incorrect, implying that $ U=V $ eventually.
	
	Now we show that $ \assump $ is an adequately permissive assumption. Again, let $ U_l $ and $ m $ be as defined earlier. Define $ X_l\coloneqq U_l\backslash U_{l-1} $ for $ 1\leq l\leq m $, and $ X_0=U_0=I $. Then every vertex $ v\in Z^* $ is in $ X_l $ for some $ l\in [0;m] $.
	
	We again prove (i)-(iii) of \cref{def:adequate assumption} separately and finally comment on the complexity of $\cobuchiTemp$.

	\begin{inparaitem}[$\blacktriangleright$]
		\noindent \item \textbf{(ii)} Implementability: 
		We again observe that the sources of the co-live edges in $ \colivegroupA $ are $ \p{1-i} $'s vertices and hence can be easily implemented by $ \p{1-i} $, by taking those edges only finitely often. Similarly, the sources of the co-live edges in $ \colivegroupS $ are $ \p{i} $'s vertices, giving the implementability of $ \Strat $ by $ \p{i} $.

		\noindent \item \textbf{(i)} Sufficiency: 	
		Consider the following strategy $ \strati $ for $ \p{i} $: at a vertex $ v\in X_0\cap \vertexi  $, she takes edge $ (v,v')\in E $ such that $ v'\in X_0  $, at a vertex $ v\in X_l\cap \vertexi $, for $ l\in [1;m] $, she plays the edges not in $ \safegroupS $ to reach $ U_{l-1} $ (such edges exist by definition of $ U_l $'s), and for all other vertices, she plays arbitrarily.
		
		Let $ v_0\in Z^*=\team{0,1}\lozenge\square I $ (from \cref{algo:coop cobuechi computation}). Let $ \stratj $ be an arbitrary strategy of $ \p{1-i} $ following $\assump$, and $ \play=v_0v_1\ldots $ be an arbitrary $ \strati\stratj $-play. Then $ \play\in \lang(\assump) $. It remains to show that $ \play \in \lang(\spec) $.

		Since $ \play\in \lang(\templatesafe(\safegroup)) $ and by definition of $ \strati $, $ v_i\in Z^* $ for all $ i $. Now suppose $ \play\not\in \lang(\spec) $, i.e. $ \inf(\play)\cap (Z^*\setminus I)\not=\emptyset $. Let $ u\in Z^*\setminus I $. Then to reach $ u $ infinitely often some edge from $ \colivegroupA$ must be taken infinitely often in $ \play $, since $ \strati $ makes the play go towards $ I $. But this contradicts the fact that $ \play\in \lang(\assump) $. Hence, $ \play\in\lang(\spec) $.

		\noindent \item \textbf{(iii)} Permissiveness: 
		Let $ \play= v_0v_1\ldots $ such that $ v_0\in Z^* $ and $ \play\in\lang(\spec) $. Suppose that $ \play\not\in \lang(\assump) $.
		
		Case 1: If $ \play\not\in\templatesafe(\safegroupA) $. Then the same argument as in the \buchi case gives a contradiction.
		
		Case 2: If $ \play\not\in\templatecolive(\colivegroupA) $, that is $ \exists (u,v)\in \colivegroupA, $ such that $ \play $ takes $ (u,v) $ infinitely often. By the definition of $ \colivegroupA$, $ v\in Z^*\setminus U_0 $, implying $ v\not\in I $, since if $ v\in I $ then it would have been in $ U_0 $. Hence, $ \play\not\in\lang(\spec) $, giving a contradiction. So $ \play\in\lang(\assump)$.

		\noindent \item \textit{Complexity analysis.}
		Very similar to that for \buchi objectives and therefore omitted.
	\end{inparaitem}
\end{proof}

\subsection{Parity games}
A \emph{parity} game is a game $\game = (\gamegraph,\spec)$ with parity objective $\spec = \paritygame(\priority)$, where
\begin{equation}\label{equ:parity}
	\paritygame(\priority) \coloneqq \bigwedge_{i\inodd [0;d]} \left(\square\lozenge \priorityset{i} \implies \bigvee_{j\ineven [i+1;d]} \square\lozenge \priorityset{j}\right),
\end{equation}
with \emph{priority set} $\priorityset{j} = \{v : \priority(v)=j\}$ for $0\leq j\leq d$ of vertices for some \emph{priority function} $\priority : V \rightarrow [0;d]$ that assigns each vertex a \emph{priority}.
An infinite play $\play$ is winning for $\spec = \paritygame(\priority)$ if the highest priority appearing infinitely often along $\play$ is even.

\begin{algorithm}[h]
	\caption{\parityTemp}
	\label{alg:compute parity assumption}
		\begin{algorithmic}[1]
			\Require $ \gamegraph=\tup{V=\vertexz\cupdot \vertexo, E}$, $\priority:V\rightarrow [0;d]$, and $i\in\{0,1\}$
			\Ensure \csm $(\assump_i,\Strat_i)$
			\State $Z^* \gets \solveParity_{0,1}(\gamegraph,\priority)$\label{algo:coop parity computation}
			\State $(\safegroupS, \safegroupA)\gets  \computeSafe(\gamegraph,Z^*,i)$\label{algo:parity:safetemp}
			\State $\gamegraph\gets \gamegraph|_{Z^*}$, $\priority\gets \priority|_{Z^*}$\label{algo:restricted game}
			\State $ (\colivegroupS,\colivegroupA,\condlivegroupS, \condlivegroupA) \gets $\textsc{ComputeSets}$ ((\gamegraph,\priority,i),\emptyset,\emptyset,\emptyset,\emptyset,\emptyset) $
			\State \Return $Z^*,\conflict,(\safegroupS,\colivegroupS,\condlivegroupS), (\safegroupA, \colivegroupA, \condlivegroupA)$ 
			
			\Statex
			\Procedure {ComputeSets}{$(\gamegraph,\priority,i),\conflict,\colivegroupS,\colivegroupA,\condlivegroupS, \condlivegroupA$}
			\State $d\gets \mathrm{max}\{l \mid \priorityset{l} \neq \emptyset\}$
			\If {$ d $ is odd}\label{algo:odd:begin}
			\State $ W_{\neg d}\gets \solveParity_{0,1}(\gamegraph|_{V\setminus \priorityset{d}},\priority) $ \label{algo:parity without d}
			\State $ (\colivegroupS, \colivegroupA)\gets (\colivegroupS,\colivegroupA) \cup \computeCoLive(\gamegraph, W_{\neg d},i)$ \label{algo:add depressed edges}
			\State $\conflict \gets \conflict\cup (V\setminus W_{\neg d})$
			\Else
			\State $ W_{d}\gets \solveBuchi_{0,1}(\gamegraph,\priorityset{d}) $, $W_{\neg d}\gets V\setminus W_{d}$\label{algo:buechi to reach d}
			
			\ForAll{odd $ l\in[0;d] $}
			\State $ (\livegroupS,\livegroupA)\gets \computeLive( \gamegraph|_{W_{d}},\priorityset{l+1}\cup \priorityset{l+3}\cdots \cup \priorityset{d},i) $\label{algo:add live groups}
			\State $ (\condlivegroupS, \condlivegroupA)\gets (\condlivegroupS \cup (W_d\cap \priorityset{l}, \livegroupS ),\condlivegroupA \cup (W_d\cap \priorityset{l}, \livegroupA )) $\label{algo:add cond live groups}
			\EndFor
			
			\EndIf
			\If {$d> 0$}
			\State $ \gamegraph\gets \gamegraph|_{W_{\neg d}} $ , $ \priorityset{0}\gets \priorityset{0}\cup \priorityset{d} $, $\priorityset{d}\gets\emptyset$\label{algo:reduce game to few color}
			\State \textsc{ComputeSets}$ ((\gamegraph,\priority,i),\conflict,\colivegroupS,\colivegroupA,\condlivegroupS, \condlivegroupA) $\label{algo:recursively compute}
			\Else 
			\State \Return $(\colivegroupS,\colivegroupA,\condlivegroupS, \condlivegroupA)$
			\EndIf
			\EndProcedure
		\end{algorithmic}
\end{algorithm}

Let us present \cref{alg:compute parity assumption} that computes assumption and strategy templates for parity games.
Since parity objectives are more complex than either \buchi or \cobuchi objectives individually, but have some similarities with both (i.e. nesting of trying to reach some priorities always eventually and to avoid some others eventually), our algorithm outputs combinations of safety, conditional group-liveness and co-liveness templates. 

\cref{alg:compute parity assumption} follows the approach in \emph{Zielonka's algorithm} \cite{ZIELONKA1998135}. The algorithm again computes the cooperative winning region, the safety templates, and restricts the graph to the cooperative winning region in lines \ref{algo:coop parity computation}-\ref{algo:restricted game}. To further understand the algorithm, we discuss two cases (visualized in \cref{fig:paritycases}) possible in a parity game and see the reason for their different treatment. Firstly, if the highest priority $ d $ occurring in the game graph is odd (see \cref{fig:paritycases} (left)), then clearly, $ \p{i} $ can not win if the play visits $ \priorityset{d} $ infinitely often. So any way to win would involve eventually staying in $ V\backslash \priorityset{d} $. Within \cref{alg:compute parity assumption}, this case is treated in lines \ref{algo:odd:begin}-\ref{algo:buechi to reach d}. Here, we remove the vertices with priority $ d $ (shaded yellow in \cref{fig:paritycases} (left)) from the graph and compute the cooperative winning region $ W_{\neg d} $ (shaded green in \cref{fig:paritycases} (left)) along with the assumption/strategy templates for the restricted game (line \ref{algo:parity without d}). As \cref{alg:compute parity assumption} already removes all vertices which are not cooperatively winning in line \ref{algo:restricted game}, there must be a way to satisfy the parity condition from vertices outside $ W_{\neg d} $. Based on the above reasoning, the only way to do so is by visiting $ W_{\neg d} $ and eventually staying there, i.e., winning a co-Büchi objective for $I:=W_{\neg d}$. The templates for states in $V\setminus W_{\neg d}$ (shaded white and yellow in \cref{fig:paritycases} (left)) are therefore computed by a call to the \cobuchi algorithm giving co-liveness templates (line \ref{algo:add depressed edges}), visualized by red arrows in \cref{fig:paritycases} (left).

Now we consider the case when $ d $ is even (see \cref{fig:paritycases} (right)) treated in lines \ref{algo:buechi to reach d}-\ref{algo:add cond live groups} within \cref{alg:compute parity assumption}. Here, one way for $ \p{i} $ to win is to visit $ \priorityset{d} $ infinitely often, giving the set $ W_d $ (shaded white and yellow in \cref{fig:paritycases} (right)) computed in line \ref{algo:buechi to reach d}. In this region, it suffices to construct templates which ensure visiting a higher even priority vertex infinitely often (i.e., $\priorityset{l+1}\cup \priorityset{l+3}\cdots \cup \priorityset{d}$), whenever vertices $\priorityset{l}$ of certain odd priority $l$ are visited infinitely often. This can be captured by conditional live-groups $(W_d\cap \priorityset{l}, \livegroup)$ (added via line \ref{algo:add cond live groups} and visualized by green arrows in \cref{fig:paritycases} (right)) where $W_d\cap \priorityset{l}$ encodes the condition of seeing a vertex with odd priority $l$ infinitely often, and $\livegroup$ is the live-group computed in line \ref{algo:add live groups} ensuring progress towards a higher even priority vertex.

\begin{figure}[h]
	\centering

\begin{tikzpicture}
	\fill[colorblindgreen!40] (0, 1.5) rectangle (3, 3);
	\draw[very thick, colorblindblue] (0, 1.5) -- (3, 1.5);
	\node at (3/2, 9/4) {$ W_{\neg d} $}; %
	
	\draw[very thick, colorblindblue] (0, 0) rectangle (4.5, 3);
	\fill[colorblindorange!30] (3, 0) rectangle (4.5, 3);
	\draw[very thick, colorblindblue] (3, 0) rectangle (4.5, 3);
	\node at (15/4, 1.5) {$ P^d $};
	
	\draw[->, red, thick] (2.75,2.5) -- (3.25,2.5);
	\draw[->, red, thick] (2.75,2) -- (3.25,2);
	\draw[->, red, thick] (4.25,1.75) -- (4.25,1.25);
	\draw[->, red, thick] (3.25,1.75) -- (3.25,1.25);
	\draw[<-, red, thick] (2.75,1) -- (3.25,1);
	\draw[<-, red, thick] (2.75,0.5) -- (3.25,0.5);
	\draw[->, red, thick] (2,1) -- (1.5,1);
	\draw[->, red, thick] (1.25,0.75) -- (0.75,0.75);
	\draw[->, red, thick] (1.8,0.5) -- (1.3,0.5);

	\draw (9/4, -0.5) node {When $ d $ is odd};
\end{tikzpicture}
\hspace{1cm}
\begin{tikzpicture}
	\tikzset{every state/.style={thick, minimum size=0.4cm}}
	\fill[colorblindgreen!40] (0, 0) rectangle (3, 1.5);
	\fill[colorblindgreen!60!colorblindorange!30] (3,0) rectangle (4.5, 1.5);
	\node at (3/2, 3/4) {$ W_{\neg d} $}; %
	
	\draw[very thick, colorblindred] (0, 0) rectangle (4.5, 3);
	\fill[colorblindorange!30] (3, 1.5) rectangle (4.5, 3);
	\draw[very thick, colorblindred] (3, 0) rectangle (4.5, 3);
	\draw[very thick, colorblindred] (0, 1.5) -- (4.5, 1.5);
	\node at (15/4, 1.5) {$ P^d $};
	
	\node[state] (q0) at (0.3, 2.25) {};
	\node[state] (q1) at (0.9, 2.25) {};
	\node[state] (q2) at (1.7, 2.25) {};
	\node[state] (q3) at (2.5, 2.25) {};
	\node[state] (q4) at (3.5, 2.25) {};
	
	\draw (0.125,2.41) node [anchor=north west] {\tiny$ 0 $};
	\draw (0.725,2.41) node [anchor=north west] {\tiny$ 1 $};
	\draw (1.525,2.41) node [anchor=north west] {\tiny$ 2 $};
	\draw (2.24,2.39) node [anchor=north west][font=\fontsize{2.5}{2.5}\selectfont]  {$d-1$};
	\draw (3.325,2.41) node [anchor=north west] {\tiny$ d $};
	\draw (1.82,2.41) node [anchor=north west] {\tiny$ \cdots $};
	
	\path[->, color={rgb, 255:red, 65; green, 117; blue, 5 },draw opacity=0.8, very thick, bend right] (q1) edge (q2);
	\path[->,color={rgb, 255:red, 65; green, 117; blue, 5 },draw opacity=0.8, very thick, bend left] (q1) edge (q2);
	\path[->, color={rgb, 255:red, 65; green, 117; blue, 5 },draw opacity=0.8, very thick, bend right] (q3) edge (q4);
	\path[->,color={rgb, 255:red, 65; green, 117; blue, 5 },draw opacity=0.8, very thick, bend left] (q3) edge (q4);
	\path[->, color={rgb, 255:red, 65; green, 117; blue, 5 },draw opacity=0.8, very thick, bend left] (q1) edge (q4);

	\draw (9/4, -0.5) node {When $ d $ is even};

\end{tikzpicture}
\caption{Visualization of the case-distinction in \cref{alg:compute parity assumption} for highest odd (left, blue) and highest even (right, red) color. Yellow and green regions indicate vertex sets $ \priorityset{d} $ and $ W_{\neg d} $, respectively, (overlapping for red). The border colors (blue for odd, and red for even) assist in case identification within \cref{fig:parityalgo}. Red and green arrows indicate co-live and live edges, respectively.}\label{fig:paritycases}
\end{figure}

The reader should note that in either case, for the vertices in $ W_{\neg d} $ (shaded green in \cref{fig:paritycases}) the only way to win is to satisfy the parity condition by visiting even priority vertices with priority lesser than $d$ infinitely often. Since vertices in $ \priorityset{d}\cap W_{\neg d}$ can not be visited infinitely often from this region (by construction), the priority of these vertices can be reduced to $ 0 $ (line \ref{algo:reduce game to few color}), allowing us to get a restricted game graph with fewer priorities. We then find the templates in the restricted graph recursively (line \ref{algo:recursively compute}). This recursion is visualized in 
\cref{fig:parityalgo} when the highest priority in the restricted graph is 3.

\begin{figure}
\centering
\begin{tikzpicture}
	\fill[colorblindorange!30] (5.5, 0) rectangle (7.5, 4); %
	\fill[colorblindgreen!40] (0, 1.5) rectangle (5.5, 4); %
	\fill[white] (0.1, 1.6) rectangle (5.4, 3.9); %
	\fill[colorblindorange!30] (4.5, 1.6) rectangle (5.4, 3.9); %
	\fill[colorblindgreen!30] (0.1, 1.6) rectangle (4.5, 3); 
	\fill[colorblindgreen!60!colorblindorange!30] (4.5, 1.6) rectangle (5.4, 3);
	\fill[white] (0.2, 1.7) rectangle (5.3, 2.9); %
	\fill[colorblindorange!30] (0.2, 1.7) rectangle (1.5, 2.9); %
	\fill[colorblindgreen!40] (1.5, 2) rectangle (5.3, 2.9); %

	\draw[very thick, colorblindred] (1.6, 2.1) rectangle (5.2, 2.8); %
	\draw[very thick, colorblindred] (0.1, 3) -- (5.4, 3); %
	\draw[very thick, colorblindred] (4.5, 1.6) -- (4.5, 3.9); %
	\draw[very thick, colorblindred] (0.1, 1.6) rectangle (5.4, 3.9); %
	\draw[very thick, colorblindblue] (1.5, 2) -- (5.3, 2); %
	\draw[very thick, colorblindblue] (1.5, 1.7) -- (1.5, 2.9); %
	\draw[very thick, colorblindblue] (0.2, 1.7) rectangle (5.3, 2.9); %
	\draw[very thick, colorblindblue] (0, 1.5) -- (5.5, 1.5); %
	\draw[very thick, colorblindblue] (5.5, 0) -- (5.5, 4); %
	\draw[very thick, colorblindblue] (0, 0) rectangle (7.5, 4); %
	
	\draw (8.5,0.5) node [anchor= west]   [align=left, colorblindblue] {$ W_{\neg 4} $};
	\draw (8.5,1.5) node [anchor= west]   [align=left, colorblindred] {$ W_{\neg 3} $};
	\draw (8.5,2.5) node [anchor= west]   [align=left, colorblindblue] {$ W_{\neg 2} $};
	\draw (8.5,3.5) node [anchor= west]   [align=left, colorblindred] {$ W_{\neg 1} $};

	\draw[->, colorblindblue, thick] (7.5,0.5) -- (8.4,0.5);
	\draw[->, colorblindred, thick, bend left] (5.4,1.8) -- (8.4,1.5);
	\draw[->, colorblindblue, thick, bend left] (5.3,2.5) -- (8.4,2.5);
	\draw[->, colorblindred, thick, bend left] (5.2,2.6) -- (8.4,3.5);

	\draw (6.5,2) node {$ P^3 $};
	\draw (5,3.5) node {$ P^2 $};
	\draw (0.75,2.25) node {$ P^1 $};
	\draw[font=\fontsize{3.5}{3.5}\selectfont] (4.9,1.85) node {$ P^2\rightarrow P^0 $};

	\path[->, color={rgb, 255:red, 65; green, 117; blue, 5 },draw opacity=0.8, very thick] (0.7,3.5) edge (1.2,3.5);
	\path[->, color={rgb, 255:red, 65; green, 117; blue, 5 },draw opacity=0.8, very thick] (2.45,3.5) edge (2.95,3.5);
	\path[->, color={rgb, 255:red, 65; green, 117; blue, 5 },draw opacity=0.8, very thick] (4.2,3.5) edge (4.7,3.5);
	\path[->, red, thick] (5.25,3.5) edge (5.75,3.5);

	\path[->, red, thick] (1.3,2.3) edge (1.8,2.3);
	\path[->, red, thick] (2.8,2.3) edge (3.2,2.3);
	\path[->, red, thick] (4.2,2.3) edge (4.7,2.3);
	\path[->, red, thick] (5.25,2.3) edge (5.75,2.3);
	\path[<-, red, thick] (1.3,1.9) edge (1.8,1.9);

	\draw[->, red, thick] (7,2.25) -- (7,1.75);
	\draw[<-, red, thick] (0.5,2.5) -- (0.5,2);
	
	\draw[<-, red, thick] (5.25,0.75) -- (5.75,0.75);
	\draw[<-, red, thick] (4.1,0.75) -- (4.6,0.75);
	\draw[<-, red, thick] (2.4,0.75) -- (2.9,0.75);
	\draw[<-, red, thick] (0.7,0.75) -- (1.2,0.75);

\end{tikzpicture}
\caption{Schematic representation of the recursion within \cref{alg:compute parity assumption} for $d=3$ using the color-coding from \cref{fig:paritycases}. Here, the only way to win is to see vertices in $\priorityset{2}$ infinitely often. }\label{fig:parityalgo}
\end{figure}
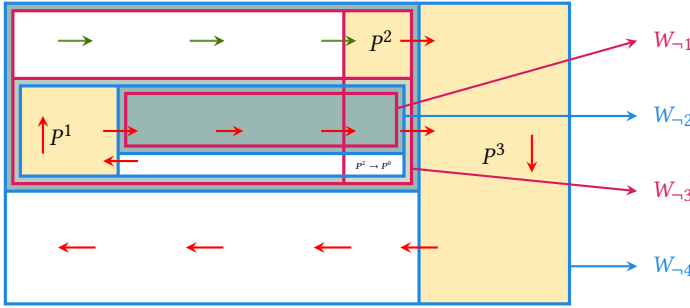

\begin{restatable}{theorem}{parityAssumption}\label{thm:parity assumption}
	Given a game graph $ \gamegraph=\tup{V=\vertexz\cupdot \vertexo, E} $ with parity objective $ \spec_i=\paritygame(\priority) $ for $ \p{i} $ where $ \priority $ is some priority function, \cref{alg:compute parity assumption} terminates in time $ \bigO(n^4) $, where $ n $ is the number of vertices in the graph. Moreover, with $ (\cdot,\cdot,\assump_i, \Strat_i):=\parityTemp(\gamegraph,\priority,i)$ the tuple $(\assump_i, \Strat_i)$ is an adequately permissive \csm for $ \p{i} $. %
\end{restatable}

\begin{proof}
	The termination of the algorithm follows from the termination of the \textsc{ComputeSets} procedure. Note that if the graph is either empty or if every vertex has priority 0, then the procedure terminates trivially. Now, in every iteration of the procedure, either at least a vertex is removed (if $ d $ is odd, then $ W_{\neg D} $ has fewer vertices) or the number of priority is reduced (if $ d $ is even, then $ G\mid _{W_{\neg D}} $ has fewer priorities). Then, after finitely many iterations, \textsc{ComputeSets} terminates.

	We prove sufficiency, implementability and permissiveness below and then analyze the complexity of \cref{alg:compute parity assumption}.
	
	\begin{inparaitem}[$\blacktriangleright$]
		\noindent \item \textbf{(ii)} Implementability: 
		We note that the assumption template $ \assump $ is implementable by $ \p{1-i} $ due to the implementability of safety, liveness and co-liveness assumptions: if for a conditional live group, the corresponding vertex set is reached infinitely often, and also the sources of live groups are visited infinitely often, $ \p{1-i} $ can choose the live group edges, since they are controlled by $ \p{1-i} $. Moreover, there won't be any conflict due to the live groups as there can be no unsafe or co-live edge that is included in a conditional live group by construction. Analogously, $ \Strat $ is implementable by $ \p{i} $.
		
		\noindent \item \textbf{(i)} Sufficiency:
		We give a strategy for $ \pz $ depending on the parity of the highest priority $ d $ occurring in the game and show that it is winning under assumption $\assump$ for all vertices in the cooperative winning region $ Z^*=\solveParity(\gamegraph,\priority) $. The strategy uses finite memory and the winning strategies for $ \po $ in subgames with \buchi (\cref{thm:Buechi assumptions}) and \cobuchi (\cref{thm:coBuechi assumptions}) objectives. 
		
		By $ \buchigame(\gamegraph, U) $, we denote the game $ (\gamegraph, \square\lozenge U) $, and by $ \cobuchigame(\gamegraph, U) $, we denote the game $ (\gamegraph, \lozenge\square U) $. We also use the definitions of $d$, $ W_d $ and $ W_{\neg d} $, as in the \cref{alg:compute parity assumption}. Consider the following strategy $ \stratz $ of $ \pz $:
		
		\begin{inparaitem}[$\triangleright$]
			\item $ d $ is \textbf{odd}: If the play is in $ V\setminus W_{\neg d} $, then $ \pz $ plays the $ \cobuchigame(\gamegraph, W_{\neg d}) $ winning strategy to eventually end up in $ W_{\neg d} $. If the play is in $ W_{\neg d} \cap Z^* $, $ \pz $ plays the recursive winning strategy for $ (G|_{W_{\neg d}}, \paritygame(\priority)) $. Otherwise, she plays arbitrarily.
			
			\item $ d $ is \textbf{even}: If the play is in $ W_d $, $ \pz $ switches its strategy among $ \buchigame(\gamegraph, W_d) $, $\buchigame(\gamegraph, W_d \cup W_{d-2})$, $\ldots$, $\buchigame(\gamegraph, W_d \cup W_{d-2} \cup \cdots W_2)$ winning strategies, i.e., for each vertex, she first uses the first strategy in the above sequence, then when that vertex is repeated, she uses the second strategy for the next move, and keeps switching to the next strategies for every move from the same vertex. If the play is in $ V\setminus W_d  \cap Z^* $, then she plays the recursive winning strategy for $ (G|_{W_{\neg d}}, \paritygame(\priority)) $, where $ \priority $ is modified again as in \cref{algo:reduce game to few color}. Otherwise, she plays arbitrarily.
		\end{inparaitem}

		We prove by induction, on the highest occurring priority $ d $, that the above constructed strategy $ \stratz $ for $ \pz $ ensures satisfying the parity objective on the original game graph if the assumption $ \assump $ is satisfied. For the base case, when $ d=0 $, the constructed strategy is trivially winning, because the only existing color is even. Now let the strategy be winning for $ d-1\geq0 $.
		
		Let $ v_0\in Z^*= \solveParity(\gamegraph,\priority)$. Let $ \strato $ be an arbitrary strategy of $ \po $ following $\assump $, and $ \play=v_0v_1\ldots $ be an arbitrary $ \stratz\strato $-play. Then $ \play\in \lang(\assump) $. We need to show that $\play$ is winning, i.e., $\play\in\lang(\Phi)$. Note that by the safety assumption and by the construction of $\stratz$, $\play$ stays in the vertex set $Z^*$. 
		
		Case 1: If $ d $ is odd, then since at vertices in $ V\setminus W_{\neg d} $, $ \pz $ plays to eventually stay in $ W_{\neg d} $, the play can not stay in $ W_{\neg d} $ without violating $ \templatecolive(\colivegroup) $. And if $ \play $ eventually stays in $ W_{\neg d} $, then by the induction hypothesis, it is winning, since $ W_{\neg d}\cap \priorityset{d}=\emptyset $.
		
		Case 2: If $ d $ is even, then if the play stays in $ W_d $ eventually, and if the play visits vertices of an odd priority $i$ infinitely often, then $ \strato $ satisfies $\computeLive(\gamegraph,\priorityset{i+1}\cup \priorityset{i+2}\cup \cdots \cup \priorityset{d},i)$ by the conditional live group assumption. Note that $ \pz $ plays the $ \buchigame(\gamegraph, \priorityset{i+1}\cup \priorityset{i+2}\cup \cdots \cup \priorityset{d})$ winning strategy for infinitely many moves from every vertex occurring in $ \play $. Since  $ \strato $ satisfies $\computeLive(\gamegraph,\priorityset{i+1}\cup \priorityset{i+2}\cup \cdots \cup \priorityset{d},i)$, after these moves as well, the play visits $ (\priorityset{i+1}\cup \priorityset{i+2}\cup \cdots \cup \priorityset{d}) $. Hence the play will visit vertices of an even color $>i$ infinitely often, implying that $ \play$ is winning. Else if $\play$ stays in $ V\setminus W_{d} $ eventually, then it is winning by induction hypothesis. This gives the sufficiency of the \csm computed by the algorithm.
		
		\noindent \item  \textbf{(iii)} Permissiveness:
		Now for the permissiveness of the \csm,  let $ \play\in \lang(\spec)$. We prove the claim by contradiction and suppose that $ \play\not\in\lang(\assump) $. 
		
		Case 1: If $ \play\not\in\templatesafe(\safegroup) $. Then some edge $ (v,v')\in \safegroup $ is taken in $ \play $. Then after reaching $ v' $, $ \play $ still satisfies the parity objective. Hence, $ v'\in Z^* $, but then $ (v,v')\not\in \safegroup $, which is a contradiction.
		
		Case 2: If $ \play\not\in\templatecondlive(\condlivegroup) $. Then for some even $ j $ and odd $i<j$, $\play$ visits $W_j\cap \priorityset{i}$ infinitely often but does not satisfy the live transition group assumption $\computeLive(\gamegraph',\priorityset{i+1}\cup \priorityset{i+2}\cdots \cup \priorityset{j}),i) $, where $\gamegraph' = \gamegraph|_{W_j}$. Due to the construction of the set $W_j$, it is easy to see that once $\play$ visits $W_j$, it can never visit $V\setminus W_j$. Hence, eventually $\play$ stays in the game $\game'$ and visits $\priorityset{i}$ infinitely often. Since $\play \in \lang(\spec)$, it also visits some vertices of some even priority $>i$ infinitely often, and hence, it satisfies $\square\lozenge(\priorityset{i+1}\cup \priorityset{i+3}\cdots \cup \priorityset{j})) $ in $ \gamegraph' $. Since $\computeLive(\gamegraph',\priorityset{i+1}\cup \priorityset{i+3}\cdots \cup \priorityset{j},i) $ is a permissive assumption for $(\gamegraph', \square\lozenge(\priorityset{i+1}\cup \priorityset{i+3}\cdots \cup \priorityset{j})) $, the play $\play$ must satisfy $\computeLive(\gamegraph',\priorityset{i+1}\cup \priorityset{i+3}\cdots \cup \priorityset{j},i) $, which contradicts the assumption.
		
		Case 3: If $ \play\not\in\templatecolive(\colivegroup) $. Then for some odd $ i $ an edge $ (u,v)\in \computeCoLive(\gamegraph,W_{\neg i},i) $ is taken infinitely often. Then the vertex $ v\in V\setminus W_{\neg i} $ is visited infinitely often. Note that $ \play $ can not be winning by visiting an even $ j>i $, since otherwise $ v $ would have been in $ \solveBuchi(\gamegraph, \priorityset{j}) $ as from $ v $ we can infinitely often see $ j $, and hence would have been removed from $ \game $ for the next recursive step. Hence, $ \play $ visits some even $ j<i $ infinitely often, i.e. $ i $ is not visited infinitely often. Then $ v $ would be in $ W_{\neg d} $, which is a contradiction.

		\noindent \item \textit{Complexity analysis.} We note that the cooperative parity game can be solved in time $ \bigO((n+m)\log d) $, where $ n, m $ and $ d $ are the number of vertices, edges and priorities respectively: consider the graph where $ pz $ owns all the vertices, find the strongly connected components in time $ \bigO(n+m) $, check which of these components have a cycle with the highest priority even by reduction to even-cycle problem \cite{valerie2001Paritywordproblem}. Then \textsc{ComputeSets} takes time $ \bigO(n^2) $ for the even case, but is dominated by $ \bigO(n^3) $ time for the odd case. For every priority, \textsc{ComputeSets} is called once, that is at most $ 2n $ calls in total. Then the total running time of the algorithm is $ \bigO((n+m)\log d+2n.(n^3)) = \bigO(n^4) $.
	\end{inparaitem}
\end{proof}

\section{Additional Material for Experimental Results}\label{app:experiments-benchmarks}
Due to space constrains we describe the benchmark generation for the experiments here. 

\smallskip
\noindent\textbf{Factory Benchmark Generator.} To generate problem instances with different computational difficulty, our benchmark generator takes four parameters to change the characteristics of the game graph: the number of columns $x$, the number of rows $y$, the number of walls $w$, and the maximum number of one-way corridors $c$. Given these parameters, the workspace of the robots is constructed as follows: first, $w$ horizontal walls, i.e., walls between two adjacent rows, are generated randomly. We ensure that there is at least one passage from every row to its adjacent rows (if this is not possible with the given $w$, then $w$ is set to the maximum possible number for the given $x$ and $y$). Next, for each passage from one row to another, we randomly designate it to be a one-way corridor. For example, if a passage is an up-corridor, then the robots can only travel in the upward direction through this passage. We ensure that the maze has at most $c$ one-way corridors. 
Given such a maze, we generate a game graph for two robots, denoted $R_1$ and $R_2$, that navigate the maze starting at the lower-left and lower-right corners of the maze, respectively. In this scenario the robots only interact explicitly via their shared workspace, i.e., possibly blocking each others way to the target.

\smallskip
\noindent\textbf{Synthesis Benchmark Generation for \cref{sec:composition}.}
We have generated $2244$ benchmarks from the games used for the Reactive Synthesis Competition (SYNTCOMP)~\cite{benchmark:syntcomp} by adding randomly generated parity objectives to given parity games. The random generator takes two parameters: game graph $\gamegraph$ and maximum priority $m$. It selects 50\% of the vertices in $\gamegraph$ randomly. These vertices are assigned priorities ranging from $0$ to $m$ (including $0$ and $m$) such that $1/m$-th of those vertices are assigned priority $0$ and $1/m$-th are assigned priority $1$ and so on. The other vertices are assigned random priorities. %

 \end{appendix}

\end{document}